\documentclass[10pt, journal]{IEEEtran}
\usepackage{epsfig,amsmath,hhline}

\usepackage{times}

\usepackage{latexsym}
\usepackage{amsmath}
\usepackage{nccbbb}
\usepackage{amsthm}
\usepackage{url}
\usepackage{dsfont}
\usepackage{psfrag}
\usepackage{epsfig}

\newtheorem{theorem}{Theorem}
\newtheorem{fact}{Fact}
\newtheorem{lemma}{Lemma}
\newtheorem{corollary}{Corollary}
\newtheorem{proposition}{Proposition}

\theoremstyle{definition}
\newtheorem{defn}{Definition}
\newtheorem{example}{Example}

\theoremstyle{plain}
\newtheorem{assumption}{Assumption}
\newtheorem{condition}{Condition}

\theoremstyle{remark}

\newcommand{\ie}{\textsl{i.e.} }

\usepackage{graphicx}
\usepackage{amssymb}
\usepackage{epstopdf}
\title{Adaptive Policies for Scheduling with Reconfiguration Delay: An End-to-End Solution for All-Optical Data Centers}

\author{Chang-Heng~Wang, \IEEEmembership{Student Member,~IEEE,}  Tara~Javidi, \IEEEmembership{Member,~IEEE,}
}

\begin{document}

\maketitle

\begin{abstract}
All-optical switching networks have been considered a promising candidate for the next generation data center networks thanks to its scalability in data bandwidth and power efficiency. However, the bufferless nature and the nonzero reconfiguration delay of optical switches remain great challenges in deploying all-optical networks. This paper considers the end-to-end scheduling for all-optical data center networks with no in-network buffer and nonzero reconfiguration delay. A framework is proposed to deal with the nonzero reconfiguration delay. The proposed approach constructs an adaptive variant of any given scheduling policy. It is shown that if a scheduling policy guarantees its schedules to have schedule weights close to the MaxWeight schedule (and thus is throughput optimal in the zero reconfiguration regime), then the throughput optimality is inherited by its adaptive variant (in any nonzero reconfiguration delay regime). As a corollary, a class of adaptive variants of the well known MaxWeight policy is shown to achieve throughput optimality without prior knowledge of the traffic load. Furthermore, through numerical simulations, the simplest such policy, namely the Adaptive MaxWeight (AMW), is shown to exhibit better delay performance than all prior work.

\end{abstract}


\section{Introduction}

Massive data centers serve as the basis of a huge variety of online services and applications nowadays. The underlying network interconnects face increasingly stringent performance requirements such as high data bandwidth and low latency. All-optical networks emerge as a promising candidate for the next generation data center networks and benefit from two technical breakthroughs: (1) the advancement of dense wavelength division multiplexing (DWDM) in optical fibers and (2) the optical switches substituting traditional electronic switches, which typically incur high cost and high power demand when supporting high data bandwidth. Due to the inherently bufferless nature of optical switches, however, the data transmission in an all-optical network will need to be conducted in an end-to-end fashion. This network topology can be viewed as a single crossbar interconnecting the end hosts, except that the full bisection bandwidth is not always guaranteed. In other words, an efficient utilization of the all-optical network depends on efficient centralized controller that can schedule the end-to-end transmissions.

The main challenge of the scheduling for optical networks, as opposed to the traditional electronic crossbar fabric scheduling, comes from the fact that optical networks typically exhibit a nonzero {\em reconfiguration delay} upon changing the circuit configuration. This reconfiguration delay comprises of two factors: (1) The reconfiguration delay of the optical switches since candidate technologies (such as binary MEMS~\cite{nistica}, ROADM~\cite{ROADM}) typically involve mechanically directing laser beams and thus require a certain time to finish a reconfiguration. (2) The time for the control plane to control/communicate with the optical switches along the intended circuit. During the circuit reconfiguration, reliable packet transmission could not be supported in the network. For practical circuit switch technologies, this reconfiguration delay is significantly longer than the link-layer inter-frame gap.  For example, the reconfiguration delay for state of the art binary MEMS is $2-20 \ \mu$s~\cite{nistica}, which is significantly larger than the inter-frame gap of 9.6 ns (for 10 Gigabit Ethernet). This nonzero reconfiguration delay motivates the need for scheduling policies that explicitly account for the reconfiguration delay.

With the presence of the reconfiguration delay, it is well known that the scheduling policy should avoid reconfiguring the circuit too frequently. Most of the existing work for scheduling policies with nonzero reconfiguration delay (\cite{adaptive_batch}, \cite{mordia}, \cite{VFMW}) tend to provision the future schedules for a certain time period (usually a long time period). These scheduling policies, classified as ``quasi-static'' policies, can perform poorly since their schedules may depend on severely out-dated queue length information. In contrast, ``dynamic'' policies determine each schedule based on the most up-to-date queue length information. An example subclass of dynamic policies is the frame-based policies such as the Fixed-Frame MaxWeight (FFMW) \cite{FFMW}, which has good delay performance if the arrival statistics is known in advance. However, the problem for frame-based policies is that the duration of the frame is set in a manner to ensure that the loss of duty cycle due to the reconfiguration delay is negligible relative to the traffic load. In other words, these policies require prior knowledge of the arrival statistics and the value of reconfiguration delay to ensure the stability of the network.

In this paper, we propose a novel class of scheduling policies called adaptive policies for scheduling with nonzero reconfiguration delay. The adaptive policies make scheduling decisions every time slot and determines both the schedule and the time to reconfigure the schedule based on the most recent queue length information. The main idea behind the adaptive policy is to keep the current schedule as long as it is ``good enough'' so that the schedule reconfiguration only occurs when it is necessary. To be more specific, the construction of an adaptive policy requires a scheduling policy $\pi$ and a sublinear function $g$. At each time slot $t$, the policy $\pi$ proposes a schedule $\mathbf{\Pi}(t)$, and the adaptive policy computes the schedule weight difference between $\mathbf{\Pi}(t)$ and the current schedule. If the weight difference is less than a threshold dependent on the function $g$, then the current schedule is considered good enough; otherwise the schedule is reconfigured to $\mathbf{\Pi}(t)$. The constructed policy is called the $g$-adaptive variant of $\pi$ where the function $g$ measures the reluctance of changing the schedule, and is thus referred to as the {\em hysteresis function}. We show in this paper that if the original policy $\pi$ has schedule weight that is close to the MaxWeight schedule (either in deterministic or expected sense), which guarantees throughput optimality under a zero reconfiguration delay, then its $g$-adaptive variant achieves throughput optimality for any fixed reconfiguration delay. Note that this stability guarantee does not require prior knowledge on the arrival statistics or the value of reconfiguration delay to stabilize the network, as opposed to the frame-based policies, such as the FFMW policy.

The rest of the paper is organized as follows. 
In the next section, the network model and the notion of stability are introduced. In section \ref{main}, we introduce the class of adaptive policies and analyze its throughput as well as certain delay bounds. We then explain the mechanism of the adaptive policies and compare them with scheduling policies in the literature qualitatively in section \ref{related}. Section \ref{simulation} gives the performance evaluation and the comparison between scheduling policies through simulations. Finally, we conclude with a summary and some future directions in section \ref{conclusion}.

\section{System Model}\label{model}
\subsection{The Network Model}

Consider a set of $N$ top of rack (ToR) switches, labeled by $\{1, 2, \dots, N\}$, which are interconnected by an optical switched network, as shown in Fig. \ref{fig_model}. Each ToR switch can serve both as a source and a destination simultaneously. We assume no buffering in the optical network, hence all the buffering occurs in the edge of the network, i.e. within the ToR switches. Each ToR switch maintains $N - 1$ edge queues (either physically or virtually), which are denoted by $Q_{ij}$, where $j \in \{1, 2, \dots, N\} \backslash \{i\}$. Packets going from the ToR switch $i$ to $j$ are enqueued in the edge queue $Q_{ij}$ before transmission.

The system considered is assumed to be time-slotted, with the time indexed as $t \in \bbbn_{+} = \{0, 1, 2, \dots \}$. Each slot duration is the transmission time of a single packet, which is assumed to be a fixed value. Let $A_{ij}(t)$ and $D_{ij}(t)$ be the number of packets arrived at and departed from queue $Q_{ij}$ at time $t$, respectively. Let $L_{ij}(t)$ be the number of packets in the edge queue $Q_{ij}$ at the beginning of the time slot $t$. For ease of notation, we write $\mathbf{A}(t) = [A_{ij}(t)], \mathbf{D}(t) = [D_{ij}(t)], \mathbf{L}(t) = [L_{ij}(t)]$, where $\mathbf{A}(t), \mathbf{D}(t), \mathbf{L}(t) \in \bbbn_+^{N \times N}$. We adopt the convention that the packet arrivals occur at the end of each time slot, and the queue dynamics is then given by
\[
L_{ij}(t+1) = L_{ij}(t) - D_{ij}(t) + A_{ij}(t)
\]

We assume the arrival processes $A_{ij}(t)$ to be independent over $i, j \in \{1, 2, \dots, N\}, i\neq j$. Each process $A_{ij}(t)$ is i.i.d. over time slots. 
We call the mean of $A_{ij}(t)$ as the traffic rate $\lambda_{ij} = \mathbb{E}\{A_{ij}(0)\}$, and define the traffic rate matrix as $\boldsymbol{\lambda} = [\lambda_{ij}] \in \bbbr^{N \times N}$.

\subsection{Schedules and Scheduling Policies}

Let $\mathbf{S}(t) \in \{0, 1\}^{N \times N}$ denote the schedule at time $t$, which indicates the optical circuits established between the ToR switches. We set $S_{ij}(t) = 1$ if an optical circuit from ToR $i$ to ToR $j$ exists at time $t$, and $S_{ij}(t) = 0$ otherwise. The feasible schedules for the network are determined by the network topology and physical constraints on simultaneous data transmissions. We let $\mathcal{S}$ denote the set of all feasible schedules, \ie $\mathbf{S}(t) \in \mathcal{S}$ for all $t$. For instance, it's common to assume at any $t$ each ToR can only transmit to at most one destination, and can only receive from at most one source, i.e. $\sum_i S_{ij}(t) \leq 1, \sum_j S_{ij}(t) \leq 1$. Under such assumption $\mathcal{S} \subset \mathcal{P}$, where $\mathcal{P}$ is the set of all permutation matrices, \ie if $\mathbf{S}(t)$ has $N$ circuit connections it is a permutation matrix. Note the permutation matrices might not all be feasible in a network; however, if all such schedules are in the feasible schedule set $\mathcal{S}$, we say the network topology is non-blocking.

\begin{figure}
\centering
\includegraphics[width = 3.0in]{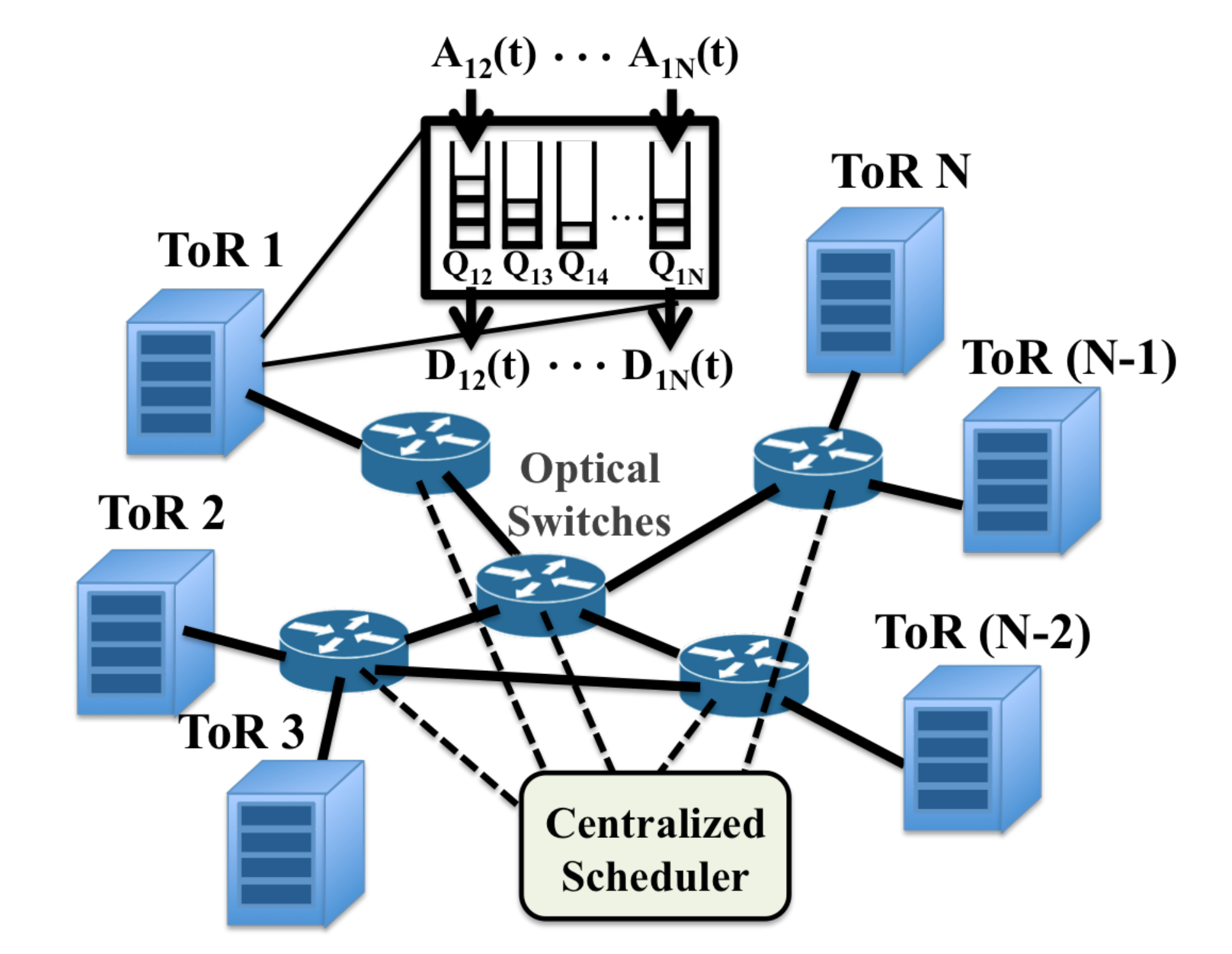}
\caption{An example of the system model}
\label{fig_model}
\end{figure}

Upon reconfiguring a schedule, the network incurs a reconfiguration delay, during which no packet could be transmitted. In this paper we focus on the effect of the reconfiguration delay on the performance of scheduling policies. We make this notion formal through the following two definitions:

\begin{defn}
Let $\{ t_k^S \}_{k=1}^{\infty}$ denote the time instances when the schedule is reconfigured. The schedule between two schedule reconfiguration time instances remains the same, \ie \[
\mathbf{S}(\tau) = \mathbf{S}(t_k^S), \ \ \forall \tau \in [t_k^S, t_{k+1}^S - 1]
\]
\end{defn}

\begin{defn}
Let $\Delta_r$ be the reconfiguration delay associated with reconfiguring the schedule of the network. During the period of schedule reconfiguration, no packet transmission could occur in the network. This means that $\forall i,j \in \{1,2, \dots N\}$, $\forall k \in \bbbn_+$, and $0 \leq \tau \leq \Delta_r$, we have $D_{ij}(t_k^S + \tau) = 0$. Note that for all other time, $D_{ij}(t) = S_{ij}(t)$ if $L_{ij}(t) > 0$.
\end{defn}

The schedule of the network is determined by a \textsl{scheduling policy}, which is defined as below.

\begin{defn}
A \textbf{scheduling policy} determines the schedules $\{ \mathbf{S}(t) \}_{t=0}^{\infty}$ where $\mathbf{S}(t) \in \mathcal{S}$ is measurable with respect to $\sigma \left\{ \{\mathbf{A}(\tau)\}_{\tau=0}^{t}, \{\mathbf{S}(\tau)\}_{\tau=0}^{t-1} \right\}$, which is the history up to time $t$.
\end{defn}

\begin{defn}
A \textbf{Markov scheduling policy} is a scheduling policy under which the schedule $\mathbf{S}(t)$ depends on the history solely through the current state $X_t = \Big( \mathbf{S}(t-1), \mathbf{L}(t) \Big)$. Restricting attention to this class of policies ensures the process $\{ X_t \}_{t=0}^{\infty}$ to be a Markov process. In this case we denote the Markov scheduling policy as $\pi$ and use the notation $\mathbf{\Pi}(t)$ to denote the choice of schedule generated by $\pi$ given the state $X_t = \Big( \mathbf{S}(t-1), \mathbf{L}(t) \Big)$.
\end{defn}

Notice that we intentionally separate the notation to let $\mathbf{S}(t)$ denote the schedules of the network and $\mathbf{\Pi}(t)$ denote the schedule generated by the Markov policy $\pi$. 

A specific scheduling policy of interest to our work is the MaxWeight policy defined below.

\begin{defn}
\label{def_weight}
Given a schedule $\mathbf{S} \in \mathcal{S}$, we define the \textbf{weight} of schedule $\mathbf{S}$ at time $t$ as 
\begin{align}
W_{\mathbf{S}}(t)
= \Big\langle \mathbf{L}(t), \mathbf{S} \Big\rangle 
= \sum_{i=1}^N\sum_{j=1}^N L_{ij}(t)S_{ij}(t).
\label{schedule_weight}
\end{align}
\end{defn}
\begin{defn}
The \textbf{MaxWeight} policy determines the schedule $\mathbf{S}^*(t)$ to be the schedule that has the maximum weight among the feasible schedules at time $t$, \ie
\[
\mathbf{\Pi}^*(t) = \arg\max_{\mathbf{S} \in \mathcal{S}} W_{\mathbf{S}}(t).
\]
We also denote the weight of the MaxWeight schedule at time~$t$ as $W^*(t) = \max\limits_{\mathbf{S} \in \mathcal{S}} \sum_{i=1}^N\sum_{j=1}^N L_{ij}(t)S_{ij}(t)$ and call it the \textbf{maximum weight} at time $t$. 
\end{defn}

\subsection{Network Stability}

\begin{defn}
The network is \textbf{strongly stable} under policy $\pi$ or a policy $\pi$ is said to strongly stabilizes the network if its queue lengths $\mathbf{L}(t)$ satisfies:
\[
\limsup\limits_{t \rightarrow \infty}  \frac{1}{t} \sum_{\tau = 0}^{t-1} \bbbe\{ ||\mathbf{L}(\tau)|| \} < \infty
\]
where $||\mathbf{L}(t)|| = \sum_{i,j=1}^N L_{ij}(t)$ is the total queue length of the system and the expectation is taken with respect to the statistics induced by a random packet arrivals and policy $\pi$. This means that a scheduling policy directly controls the statistics of queue occupancies and hence the per packet delay.
\end{defn}

With the notion of the strong stability, we may then define the admissible arrival traffic and the throughput optimality of a scheduling policy as follows:

\begin{defn}
The arrival traffic $\mathbf{A}(t)$ is \textbf{admissible} if there exists a scheduling policy which strongly stabilizes the system. The traffic rate matrix $\boldsymbol{\lambda}$ is admissible if $\mathbf{A}(t)$ is admissible.
\end{defn}

\begin{defn}
The \textbf{capacity region} is the set of all admissible traffic rate matrices and is denoted as $\boldsymbol{\Lambda}$.
\end{defn}

We may then define the traffic load as follows:
\begin{defn}
The \textbf{load} of the traffic rate matrix $\boldsymbol{\lambda}$ is defined as
\[
\rho(\boldsymbol{\lambda}) = \inf\{r: \boldsymbol{\lambda} \in r\bar{\boldsymbol{\Lambda}}, \ 0 < r < 1 \}
\]
where $\bar{\boldsymbol{\Lambda}}$ is the closure of $\boldsymbol{\Lambda}$. We shall use $\rho$ instead of $\rho(\boldsymbol{\lambda})$ whenever there is no confusion.
\end{defn}

\begin{defn}
A scheduling policy achieves \textbf{throughput optimality} if it strongly stabilizes the network under any admissible arrival traffic.
\end{defn}

When $\Delta_r = 0$, several scheduling policies (e.g. \cite{Tassiulas}, \cite{Tassiulas_random}, \cite{randomized}) has been shown to achieve throughput optimality in the literature. However, in the regime of $\Delta_r > 0$, these scheduling policies typically loses the throughput optimality guarantee since they do not address the fact that each schedule reconfiguration would result in a significant reduction in the duty cycle and hence decrease the utility of the network. The challenge for scheduling in the $\Delta_r > 0$ regime is that both the quality of the schedules and the rate of schedule reconfiguration affects the performance of a scheduling policy. It is not hard to see that there is a tradeoff between these two factors: Reducing the rate of reconfiguration means that the network is forced to stick with a schedule longer and loses the chance to use a better schedule; on the other hand, pursuing a better schedule most of the time inevitably increases the rate of reconfiguration and thus the incurred overhead. Therefore, a good scheduling policy in the $\Delta_r > 0$ regime must strive to achive the balance between these two factors.

\section{Main Results}\label{main}

In this section, we introduce a broad class of adaptive scheduling policies that balance the rate of schedule reconfiguration and the quality of the schedules without prior knowledge of the arrival traffic. We first introduce a general approach that could transform any scheduling policy to an adaptive variant of the original policy and give some example policies generated by this approach. We then introduce several adaptive policies that could achieve throughput optimality in the $\Delta_r > 0$ regime. In particular, we primarily focus on scheduling policies that have the schedule weight close to the MaxWeight policy (either in deterministic sense or in expectation). These scheduling policies have been shown in the literature (e.g. \cite{MaxWeight-f}) to achieve throughput optimality when $\Delta_r = 0$. In this paper, we show that under mild conditions, the adaptive variants of these policies achieve throughput optimality under any fixed reconfiguration delay $\Delta_r > 0$.

\subsection{Adaptive Policies}
\label{adaptive_policies}

Given a Markov scheduling policy $\pi$ and current state $X_t$, let $\mathbf{\Pi}(t) = \pi(X_t)$ denote the schedule generated by $\pi$ at time $t$. We then define the following:
\begin{itemize}
\item $W^{\pi}(t) = W_{\mathbf{\Pi}(t)}(t)$ is the weight of the schedule $\mathbf{\Pi}(t)$ at time $t$
\item $\tilde{W}(t) = W_{\mathbf{S}(t-1)}(t)$ is the weight of the previous schedule $\mathbf{S}(t-1)$ at time $t$
\end{itemize}

Now note that at any given time $t$, $\Delta W(t) =  W^{\pi}(t) - \tilde{W}(t)$ measures the potential improvement (in terms of schedule weight (\ref{schedule_weight})) associated with following policy $\pi$ instead of sticking with the previously used schedule. Since each schedule change results in a loss in duty cycle, our proposed class of adaptive policies show some inertia against frequent schedule reconfiguration. More precisely, let us define a \textbf{hysteresis function} $g$ where $g: \bbbr \rightarrow \bbbr$ is a nonnegative, continuous, strictly increasing, and sublinear (\ie $\lim\limits_{x \rightarrow \infty} \frac{g(x)}{x} = 0$) function.

Our proposed $g$-adaptive Markov policy $\pi^g$ uses the new schedule $\mathbf{\Pi}(t)$ only if $\Delta W(t) > g(W^{\pi}(t))$. In other words,
\begin{align*}
\mathbf{\Pi}^{g}(t) 
&= \pi^g \Big( \mathbf{S}(t-1), \mathbf{L}(t) \Big)  \\
&= \left\{
\begin{array}{ll}
\mathbf{S}(t-1) & \mbox{if } \Delta W(t) \leq g(W^{\pi}(t)) \\
\mathbf{\Pi}(t) & \mbox{if } \Delta W(t) > g(W^{\pi}(t))
\end{array}
\right. 
\end{align*}
Note that the proposed $g$-adaptive Markov policy could be constructed from any Markov scheduling policy. Given the Markov policy $\pi$, we call $\pi^g$ the $g$-adaptive variant of $\pi$.

The intuition behind this construction is that the $g$-adaptive policy holds on to the last schedule as long as it is ``good enough'' relative to the current schedule generated by $\pi$, $\Pi(t)$. The sublinearity of the function $g(\cdot)$ is a technical assumption used to achieve the throughput optimality, which would become clear in the analysis given in the next subsection.

We now give some examples for possible combinations of the function $g(\cdot)$ and the scheduling policy $\pi$.

\begin{example}{(the Adaptive MaxWeight policy \cite{Infocom15})}\\
\indent Let the scheduling policy $\pi$ be the MaxWeight policy, and the function $g(\cdot)$ takes the form $g(x) = (1-\gamma)x^{1-\delta}$, where $\gamma \in (0, 1), \delta \in (0,1)$. Then the $g$-adaptive variant of $\pi$ is called the Adaptive MaxWeight, as introduced in \cite{Infocom15}. \\
\indent The Adaptive MaxWeight policy computes the weight difference between the MaxWeight schedule and its last schedule $\Delta W = W^*(t) - \tilde{W}^*(t)$ to the threshold $g(W^*(t)) = (1-\gamma) (W^*(t))^{1-\delta}$ at each time slot. It reconfigures to the MaxWeight schedule if the difference is above the threshold, and keeps the current schedule if otherwise.
\label{adaptive_maxweight}
\end{example}

\begin{example}{(the $g$-adaptive variant of the Pipelined MaxWeight policy)}\\
\indent Under the pipelined MaxWeight policy \cite{1-APRX}, the scheduler determines the schedule at time $t$ to be the schedule maximizing the weight at time $t-K$, for some fixed scalar $K < \infty$. Intuitively, one can think of the pipelined MaxWeight as a scheduler that initiates MaxWeight computation at each time slot but obtains and enforces it only $K$ time slots later. Therefore, the schedule at time $t$ is the MaxWeight schedule based on the queue length information at time $t-K$, \ie $\mathbf{L}(t-K)$. 
The selection of the function $g(\cdot)$ could be any continuous, strictly increasing, and sublinear function, e.g. $g(x) = (1-\gamma)x^{1-\delta}$ as in Example~\ref{adaptive_maxweight}.
\label{adaptive_pipeline_maxweight}
\end{example}

Since the MaxWeight scheduling policy is well known for its high computation complexity, in the literature several lower complexity policies have also been proposed with good stability conditions when $\Delta_r = 0$. We may consider the adaptive variant of these policies as well.

\begin{example}{(the $g$-adaptive variant of the Tassiulas random policy)} \\
\indent The Tassiulas Random policy \cite{Tassiulas_random} utilizes random schedule selection and memory to determine the schedule. It compares the weight between the last schedule and a randomly selected schedule (according to an arbitrary distribution on the feasible schedule $\mathcal{S}$, say uniformly random). Let $\mathbf{Z}(t)$ be the randomly selected schedule at time $t$, then the schedule determined by the Tassiulas random policy is given by
\begin{align*}
\mathbf{\Pi}_T(t)= \left\{
\begin{array}{ll}
\mathbf{\Pi}_T(t-1) & \mbox{if }  W_{\mathbf{\Pi}_T(t-1)}(t) \geq W_{\mathbf{Z}(t)}(t) \\
\mathbf{Z}(t) & \mbox{otherwise } 
\end{array}
\right.
\end{align*}
\label{adaptive_tassiulas}
\end{example}

\begin{example}{(the $g$-adaptive variant of the Hamiltonian policy)} \\
\indent The Hamiltonian policy \cite{randomized} utilizes the Hamiltonian walk on the set of permutation matrices and memory to determine the schedule. It compares the schedule weight between the last schedule and the  schedule on the Hamiltonian path and select the schedule with higher weight. Specifically, let $\mathbf{H}(t)$ be the schedule on the Hamiltonian path at time $t$, then the schedule determined by the Hamiltonian policy is given by
\begin{align*}
\mathbf{\Pi}_H(t)= \left\{
\begin{array}{ll}
\mathbf{\Pi}_H(t-1) & \mbox{if }  W_{\mathbf{\Pi}_H(t-1)}(t) \geq W_{\mathbf{Z}(t)}(t) \\
\mathbf{H}(t) & \mbox{otherwise } 
\end{array}
\right.
\end{align*}
\label{adaptive_hamilton}
\end{example}

\begin{example}{(the $g$-adaptive variant of the Maximum Size scheduling policy)} \\
\indent A maximum size scheduling policy \cite{Maxsize} selects a schedule that has the maximum number of nonempty queues. It can be view as a variation of the MaxWeight policy except that each nonempty queue has weight one, and each empty queue has weight zero. In particular, the schedule is selected as
\[
\mathbf{\Pi}_{MS}(t) = \arg\max_{\mathbf{S} \in \mathcal{S}} \sum_{i=1}^N\sum_{j=1}^N S_{ij}(t) \mathds{1}_{\{L_{ij}(t) > 0\}}
\]
where $\mathds{1}_{\{L_{ij}(t) > 0\}}$ is the indicator function for the event $\{L_{ij}(t) > 0\}$. 
\label{adaptive_maxsize}
\end{example}

In the next subsection we discuss conditions required for an adaptive variant to achieve throughput optimality under $\Delta_r > 0$. Notice that the maximum size policy does not achieve throughput optimality even in the $\Delta_r = 0$ regime \cite{Maxsize}, therefore we do not consider the throughput analysis of example~\ref{adaptive_maxsize}.

\subsection{Drift Analysis}

Throughout this subsection, we impose the following assumption on the arrival traffic:
\begin{assumption}
Assume the arrival at each queue is bounded from above, \ie $\mathbf{A}_{ij}(t) \leq A_{max} < \infty, \forall i,j \in \{1,2, \dots, N\}, \forall t$.
\label{assumption_traffic}
\end{assumption}

We start with a class of scheduling policies that guarantee bounded schedule weight difference to the MaxWeight schedule at all times. 

\begin{condition}
Suppose there exists a constant $G < \infty$ such that the schedule weight of the scheduling policy $\pi$, $W^{\pi}(t) = W_{\mathbf{\Pi}(t)}(t)$, satisfies
\begin{align*}
W^{\pi}(t) \geq W^*(t) - G, \ \ \forall t
\end{align*}
\label{strict_bound}
where $W^*(t)$ is the maximum weight, the schedule weight of the MaxWeight schedule at time $t$ as introduced in Section \ref{model}.
\end{condition}
Given a scheduling policy $\pi$ that satisfies Condition~\ref{strict_bound},the following Theorem establishes that the $g$-adaptive variant of $\pi$ achieves the throughput optimality.

\begin{theorem}
\label{Adaptive_Approximation}
Given any reconfiguration delay $\Delta_r > 0$. Assume the traffic is admissible and satisfies assumption~\ref{assumption_traffic}. Suppose a scheduling policy $\pi$ satisfies Condition~\ref{strict_bound}, then the $g$-adaptive variant of $\pi$ is throughput optimal.

Moreover, assume further that $g(\cdot)$ is concave. Let $G < \infty$ be the constant in condition \ref{strict_bound}. Let $M = g^{-1}(G + N(A_{\max}+1)T) + NT$, then the $g$-adaptive variant of $\pi$ guarantees the mean queue length to satisfy
\begin{align}
\lim_{t \rightarrow \infty} \bbbe  \left[ ||\mathbf{L}(t)|| \right] \leq \inf_{T > \frac{\Delta_r}{1-\rho}} \tilde{L}_T
\end{align}
where $\tilde{L}_T$ is the fixed point solution to the equation: 
\begin{align*}
\tilde{L}_T =& \frac{N}{1-\frac{\Delta_r}{T}-\rho} 
\Big\{  g(\tilde{L}_T + NA_{\max}T) + G \\
&+ N (T + A_{\max} \Delta_r) 
 + M(1 + NA_{\max}) + \frac{N^2 A_{\max}^2}{2} \Big\}
\end{align*}
\end{theorem}

The proof of Theorem \ref{Adaptive_Approximation} is given in Appendix. The proof utilizes the Foster-Lyapunov Theorem and consists of two main components. The first one shows that the schedule of the adaptive policy $\pi^g$ has weight difference to the MaxWeight schedule bounded by a sublinear function of the maximum weight. This potentially gives the negative Lyapunov drift, as similarly argued in \cite{1-APRX}. The second part is that the rate of schedule reconfiguration becomes smaller as the queue lengths become larger. With this property, we show that the overhead incurred by the schedule reconfiguration delay becomes arbitrarily small when the total queue lengths $||\mathbf{L}(t)||$ increases. This suffices to give the guarantee of negative expected drift when $||\mathbf{L}(t)||$ is large and thus guarantee the stability.

With Theorem \ref{Adaptive_Approximation}, we are now ready to show the throughput optimality of some example adaptive policies given in the previous subsection.

\begin{corollary}
\label{coro1}
The adaptive policies in examples~\ref{adaptive_maxweight} and \ref{adaptive_pipeline_maxweight} achieve throughput optimality.
\end{corollary}
\begin{proof}
For example~\ref{adaptive_maxweight}, since $\pi$ is the MaxWeight policy, Condition \ref{strict_bound} is satisfied by definition, with $G = 0$. As for example~\ref{adaptive_pipeline_maxweight}, Condition \ref{strict_bound} is satisfied with $G = N(A_{\max}+1)K$. To see this, note that at most one packet could depart from each queue at each time slot, and hence 
\begin{align}
\Big\langle \mathbf{L}(t), \mathbf{\Pi}(t) \Big\rangle &\geq \Big\langle \mathbf{L}(t-K), \mathbf{\Pi}(t) \Big\rangle - NK \nonumber \\
& = \Big\langle \mathbf{L}(t-K), \mathbf{\Pi}^*(t-K) \Big\rangle - NK.
\label{coro1_1}
\end{align}
Since the arrivals are bounded by $A_{\max}$, we also have
\begin{align}
\Big\langle \mathbf{L}(t), \mathbf{\Pi}^*(t) \Big\rangle \leq \Big\langle \mathbf{L}(t-K), \mathbf{\Pi}^*(t-K) \Big\rangle + NA_{\max}K.
\label{coro1_2}
\end{align}

From (\ref{coro1_1}) and (\ref{coro1_2}), we have $\Big\langle \mathbf{L}(t), \mathbf{\Pi}(t) \Big\rangle \geq \Big\langle \mathbf{L}(t), \mathbf{\Pi}^*(t) \Big\rangle - N(A_{\max}+1)K$ and thus Condition 1 is satisfied.
\end{proof}

\begin{corollary}
\label{coro2}
The $g$-adaptive variant of the Hamiltonian policy given in example~\ref{adaptive_hamilton} achieves throughput optimality.
\end{corollary}

\begin{proof}
In \cite{randomized} the Hamiltonian policy is shown to satisfy Condition \ref{strict_bound} with $G = 2N (N!)$, hence by Theorem \ref{Adaptive_Approximation} it achieves throughput optimality.
\end{proof}

As we saw in example~\ref{adaptive_tassiulas}, some low complexity scheduling policies utilize the random schedule selection and memory to approximate the MaxWeight policy (e.g. \cite{Tassiulas_random}, \cite{randomized}). For a random scheduling policy $\pi$, the schedule $\mathbf{\Pi}(t)$ is a sequence of random schedules from the feasible schedule set $\mathcal{S}$. These policies guarantee bounded weight differences to the MaxWeight schedule in the expected sense: 

\begin{condition}
Consider a Markov scheduling policy $\pi$ with weight $W^{\pi}(t)$ that satisfies the property:
\begin{align}
\bbbe^{\pi} \left[ W^{\pi}(t) | \mathbf{L}(t) \right] \geq W^*(t) - G, \ \ \forall t
\end{align}
where $G < \infty$ is a constant, and the expectation $\bbbe^{\pi}$ is taken with respect to the policy $\pi$.
\label{expected_bound}
\end{condition}

We consider the throughput optimality of the adaptive variants of policies satisfying the above condition as well. 

\begin{theorem}
\label{Adaptive_Expected}
Given any reconfiguration delay $\Delta_r > 0$. Assume the traffic is admissible and satisfies assumption \ref{assumption_traffic}. For any scheduling policy $\pi$ that satisfies Condition \ref{expected_bound}, the $g$-adaptive version of $\pi$ achieves throughput optimality.
\end{theorem}

With Theorem \ref{Adaptive_Expected}, we can extend Corollaries~\ref{coro1} and \ref{coro2} to the $g$-adaptive variant of the Tassiulas random policy in example~\ref{adaptive_tassiulas}:

\begin{corollary}
The $g$-adaptive variant of the Tassiulas random policy given in example~\ref{adaptive_tassiulas} achieves throughput optimality.
\end{corollary}

\begin{proof}
Let $\Pr \{ \mathbf{Z}(t) = \mathbf{\Pi}^*(t) \} \geq \epsilon > 0$ for all $t$ and for some $\epsilon$. In \cite{Tassiulas_random_analysis} the Tassiulas random policy is shown to satisfy Condition \ref{expected_bound} with $G = 2N \frac{1-\epsilon}{\epsilon}$, hence by Theorem \ref{Adaptive_Expected} it achieves throughput optimality.
\end{proof}

\subsection{Generalization}

In the previous subsection, we established the throughput optimality of the $g$-adaptive variant of any scheduling policy that has schedule weight close to the MaxWeight policy, either in deterministic or expected sense. In the literature the idea of MaxWeight policy has been generalized to a broader class of MaxWeight-$f$ policy by extending the definition of schedule weight to a more general $f$-weight. This extended class of policies have been shown to achieve throughput optimality. In this subsection we briefly discuss the criteria for the adaptive variants of these policies to achieve throughput optimality.

We first introduce the MaxWeight-$f$ policy. Given a function $f: \bbbr \rightarrow \bbbr$ which is strictly increasing, and $f(0) = 0$, we define the $f$-weight of a schedule $\mathbf{S} \in \mathcal{S}$ at time $t$ as 
\[
W_{f, \mathbf{S}}(t)
= \Big\langle f(\mathbf{L}(t)), \mathbf{S} \Big\rangle 
= \sum_{i=1}^N\sum_{j=1}^N f(L_{ij}(t))S_{ij}(t).
\]
The MaxWeight-$f$ policy is a generalization of the MaxWeight policy in that it selects the schedule $\mathbf{\Pi}_f^*(t)$ that has the maximum $f$-weight among the feasible schedules, \ie
\[
\mathbf{\Pi}_f^*(t) = \arg\max_{\mathbf{S} \in \mathcal{S}} \sum_{i=1}^N\sum_{j=1}^N f(L_{ij}(t))S_{ij}(t).
\]
We denote the $f$-weight of the MaxWeight-$f$ schedule at time $t$ as $W_f^*(t) = \arg\max\limits_{\mathbf{S} \in \mathcal{S}} \sum_{i=1}^N\sum_{j=1}^N f(L_{ij}(t))S_{ij}(t)$.

Recall that we construct the $g$-adaptive policies using schedule weight comparison, here we generalize the approach by using the more general $f$-weight comparison, and refer to it as the $(g,f)$-adaptive policies. Formally, given a Markov scheduling policy $\pi$ and current state $X_t$, let $\mathbf{\Pi}(t) = \pi(X_t)$ denote the schedule generated by $\pi$ at time $t$. Let $f(\cdot)$ be the weight function. We then define the following:
\begin{itemize}
\item $W_f^{\pi}(t) = W_{f, \mathbf{\Pi}(t)}(t)$ is the $f$-weight of the schedule $\mathbf{\Pi}(t)$ at time $t$
\item $\tilde{W}_f(t) = W_{f, \mathbf{S}(t-1)}(t)$ is the $f$-weight of the previous schedule $\mathbf{S}(t-1)$ at time $t$
\end{itemize}

We then define the \textbf{$(g,f)-$adaptive variant} of $\pi$, denoted as $\pi^{g,f}$, with the schedule $\mathbf{\Pi}^{g,f}(t)$ determined as follows:
\begin{align}
\mathbf{\Pi}^{g,f}(t)= \left\{
\begin{array}{ll}
\mathbf{S}(t-1) & \mbox{if } \Delta W_f(t) \leq g(W_f^{\pi}(t)) \\
\mathbf{\Pi}(t) & \mbox{if } \Delta W_f(t) > g(W_f^{\pi}(t))
\end{array}
\right.
\end{align}
where $\Delta W_f(t) =  W_f^{\pi}(t) - \tilde{W}_f(t)$ is the $f$-weight difference between the schedule $\mathbf{\Pi}(t)$ and the previous schedule $\mathbf{S}(t-1)$.

An example of the $(g,f)$-adaptive policies is given below.

\begin{example}{(the $(g,f)$-adaptive MaxWeight-$\alpha$ policy, $\alpha > 0$)} \\
\indent For any fixed $\alpha > 0$, let the weight function be $f(x) = x^{\alpha}$, then the MaxWeight-$f$ policy is also referred as MaxWeight-$\alpha$ policy in the literature. Under the MaxWeight-$\alpha$ policy, the schedule at time $t$ is given by
\[
\mathbf{\Pi}_{MW\alpha}(t) = \arg\max_{\mathbf{S} \in \mathcal{S}} \sum_{i=1}^N\sum_{j=1}^N L_{ij}^{\alpha}(t)S_{ij}(t)
\]
Recall the hysteresis function $g(\cdot)$ is a strictly increasing and sublinear function. While to ensure the throughput optimality, further restrictions on $g(\cdot)$ would be required, which may be dependent on the weight function $f(\cdot)$. This is shown in the following proposition.
\label{adaptive_maxweight_alpha}
\end{example}

\begin{proposition}
Given the weight function $f(x) = x^{\alpha}$ with $\alpha > 0$. Suppose the sublinear function $g(\cdot)$ satisfies $\lim\limits_{x \rightarrow \infty} \frac{x^{\alpha-1}}{g(x^{\alpha})} = 0$, then the $(g,f)$-adaptive variant of the MaxWeight-$\alpha$ in example~\ref{adaptive_maxweight_alpha} achieves throughput optimality. 
\end{proposition}

The proof requires modifying the proof of Theorem \ref{Adaptive_Approximation} by taking the Lyapunov function as $V \left( \mathbf{L} \right) = \sum_{i,j=1}^N F(L_{ij})$ where $F(x) = \int_0^x f(y) dy$, and is omitted here for brevity.

\section{Discussion}\label{related}
In this section, we first give a brief explanation on how the adaptive policies achieve the throughput optimality. We then categorize policies dealing with nonzero reconfiguration delay into two classes, namely quasi-static and dynamic policies.

\begin{figure}
\centering
\includegraphics[width = 3.0in]{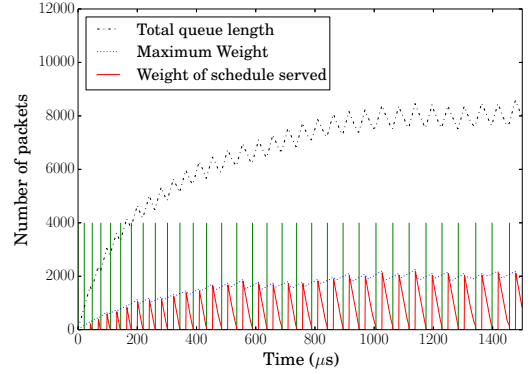}
\caption{Sample paths of the total queue length, maximum weight, and the schedule weight, under the Adaptive MaxWeight (AMW) policy. The number of ToRs is $N = 8$, the traffic load is $\rho = 0.6$, and the parameters of the AMW policy are $\delta = 0.01, \gamma = 0.1$. The green lines mark the schedule reconfiguration time instances $t_k^S$.}
\label{sample_path}
\end{figure}

\subsection{Queue Dynamics under the Adaptive MaxWeight Policy}
Fig. \ref{sample_path} illustrates the sample paths of the total queue lengths, the maximum weight, and the schedule weight of a network under the Adaptive MaxWeight policy. The green vertical lines mark the times of schedule reconfiguration, $t_k^S$, and the schedule weight is set as zero during the reconfiguration delay (time period of length $\Delta_r$ following each schedule reconfiguration instance $t_k^S$). We first observe that the schedule duration increases as the total queue length increases. This is an essential component in achieving the network stability: let $T$ be the mean schedule duration, then in order to ensure stability under $\Delta_r > 0$, the rate of schedule reconfiguration must satisfy $1 - \frac{\Delta_r}{T} > \rho$. The schedule duration increases with the queue length until this condition is satisfied.

A somewhat more interesting observation is in the mechanism of this schedule duration adjustment. Recall that in the construction of $(g, f)$-adaptive policy, it requires no explicit adjustment of the schedule duration, \ie the duration of a schedule is not explicitly set at the time it is reconfigured. Instead, through the schedule determination by weight comparison at each time slot, the schedule duration is implicitly ``adapted'' to the appropriate value. After each schedule reconfiguration, the schedule weight begins from the maximum weight and decreases as the network is serving the queues associate with the schedule. When the total queue length is larger (and thus the maximum weight is larger since $W^*(t) \geq \frac{1}{N} ||\mathbf{L}(t)||$), the threshold is larger and it takes longer for the weight difference between the maximum weight and the schedule weight to surpass the threshold. This property may be characterized by lemma \ref{lemma_frequency} shown in Appendix \ref{theorem_adaptive_approximation_proof}. 

Note that this mechanism differs from other scheduling policies in the literature that explicitly adjust the schedule duration (\cite{adaptive_batch}, \cite{mordia}, \cite{VFMW}). We give a brief introduction of these policies and categorize them based on this schedule provisioning behavior in the next subsection. The performance comparison of these policies is then evaluated through simulations and presented in the next section.

\subsection{Quasi-static v.s. Dynamic Policies}

For scheduling policies accounting for nonzero reconfiguration delay, we classify them into two categories: ``quasi-static scheduling'' and ``dynamic scheduling.'' Quasi-static scheduling policies, also refered to as ``batch scheduling,''~\cite{adaptive_batch} select a series of schedules based on a single schedule computation process, as shown in Figure~\ref{time} (a). We argue that under these policies, the generated schedules may depend on very out-dated information, especially for schedules employed later in a batch. This is the source of significant performance degradation. In contrast, under dynamic scheduling policies, each schedule is generated based on the most up-to-date queue information, as shown in Figure~\ref{time} (b). We now describe several example policies for each class and discuss their performance.

In \cite{adaptive_batch}, batch scheduling policies are further classified as fixed batch scheduling (FBS) policies or adaptive batch scheduling (ABS) policies depending on how to determine the time instances for schedule computation. If the time between two schedule computation is a fixed duration then it is a FBS policy, otherwise it is a ABS policy. In general, a FBS policy does not guarantee stability unless the number of schedules employed within a batch is restricted to avoid too frequent schedule reconfiguration. The traffic matrix scheduling (TMS) policy \cite{mordia} is a special case of FBS policy which could guarantee stability if the traffic load is known in advance. Under the TMS policy, the schedule computation occurs at $t_k = kW$ for some integer parameter $W < \infty$ and $k = 0,1,\dots$ The TMS policy utilizes the Birkhoff von-Neumann (BvN) decomposition \cite{bvn_alg} to determine the schedules in the following $W$ time slots. \footnote{The BvN decomposition is based on the BvN Theorem \cite{bvn_theorem} that every doubly stochastic matrix could be decomposed as a convex combination of permutation matrices} More specifically, the scheduler takes the queue length information $\mathbf{L}(t_k)$ and scales it to a doubly stochastic matrix $\mathbf{B}(t_k)$ which indicates the relative service requirement in the following $W$ slots and performs the BvN decomposition on $\mathbf{B}(t_k)$ as $\mathbf{B}(t_k) = \sum_{i=1}^Q \alpha_i \mathbf{P}_i$. Note the number of terms $Q$ in the decomposition may vary (while $Q \leq N^2-2N+2$). In practice, the parameter $Q$ is set as a fixed number to avoid excessive schedule changes and $Q$ largest weighted schedules are chosen. If the arrival traffic load is known a priori, an appropriate choice of parameters ensures the stability of TMS.

\begin{figure}
\centering
\includegraphics[width = 2.5in]{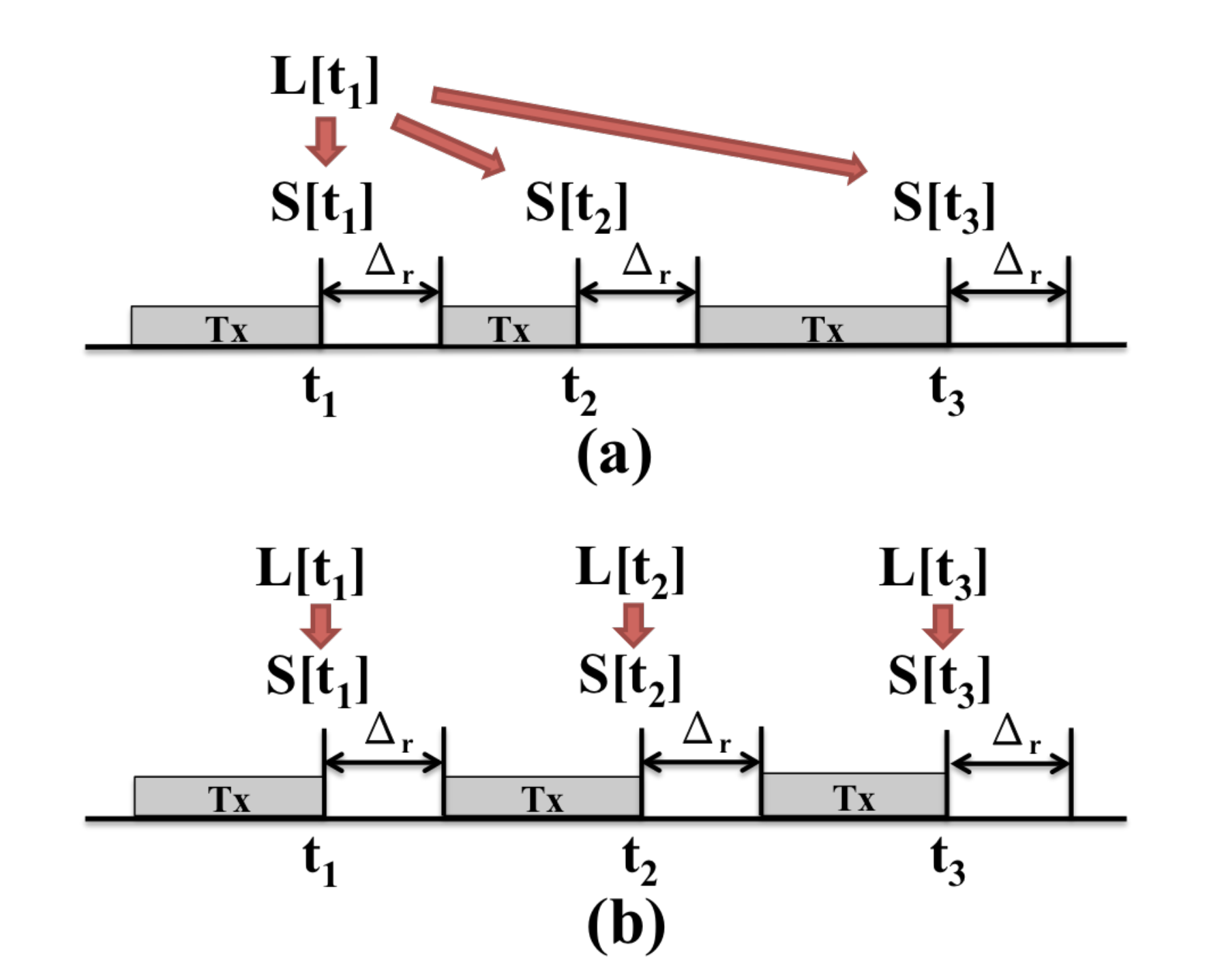}
\caption{Timing diagram for different scheduling strategies. (a) Quasi-static policies: Series of schedules determined in a single schedule computation, and some schedules could depend on out-dated queue information when being deployed.
(b) Dynamic policies: Each schedule is computed based on the most up-to-date edge queue information.}
\label{time}
\end{figure}

The ABS policy proposed in \cite{adaptive_batch} determines each schedule computation time as the time the packets from last batch are cleared. The ABS guarantees rate stability \cite{adaptive_batch} which is a weaker notion of stability compared to the strong stability considered in this paper. In fact, the expected queue length may be unbounded under the ABS policy, hence it is not considered in the performance comparison in this work. On the other hand, the Variable Frame MaxWeight (VFMW) policy proposed in~\cite{VFMW} can be viewed as a variant to the ABS policy. The VFMW policy selects only one schedule for each batch, and the batch duration is a function of the batch size (instead of being the time that the batch is cleared). While the VFMW policy is shown in~\cite{VFMW} to be throughput optimal, it has poor delay performance since it selects and fixes the schedule duration disregarding the arrivals in the schedule duration, which is a similar problem to other quasi-static policies.

In contrast to the quasi-static policies, under the dynamic scheduling policies, the schedule being employed at each schedule reconfiguration instance is based on the most up-to-date queue length information. Frame-based policies, such as the Fixed Frame MaxWeight (FFMW) policy~\cite{FFMW}, are examples for dynamic policies. The FFMW policy selects a fixed period $T$ and sets the schedule reconfiguration times at $t_n = nT, n = 0,1,2,\dots$ At each $t_n$, the schedule is selected as the MaxWeight schedule at time $t_n$, therefore each schedule is based on the most up-to-date queue information and result in an improved delay performance. Unfortunately, like the TMS policy, however, the FFMW requires prior knowledge of the traffic load to guarantee the network stability which is similar to the TMS policy. Note that by the definition given, our proposed $(g,f)$-adaptive policies are other examples of dynamic policies. In contrast to the FFMW, however, our proposed $(g,f)$-adaptive policies achieve throughput optimality without prior knowledge of the traffic load. Moreover, the $(g,f)$-adaptive policies also have good delay performance, as would be shown through simulations in the next section.

\section{Simulation} \label{simulation}
\begin{figure}
\includegraphics[width = 3.2in]{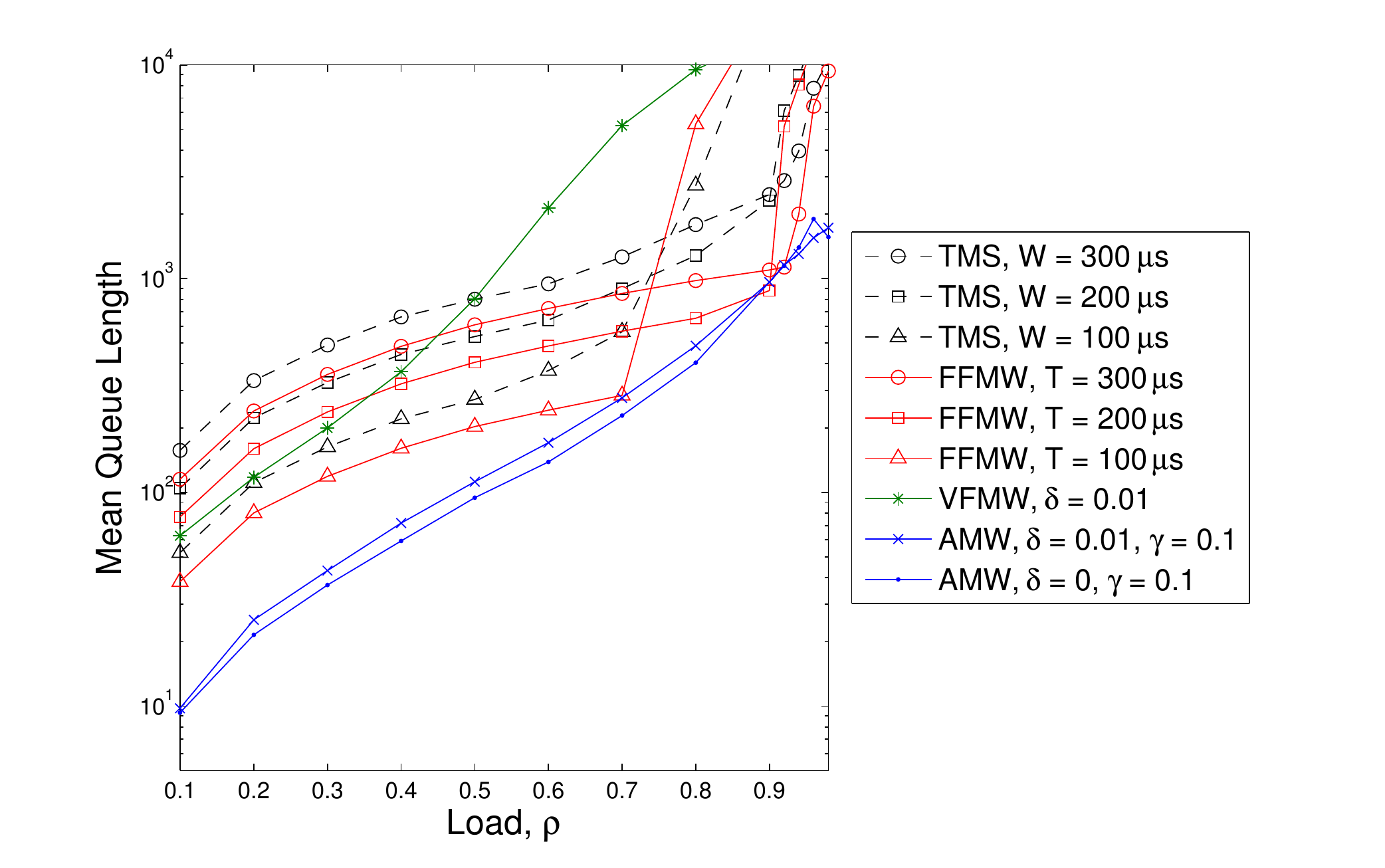}
\caption{Mean queue length versus traffic load $\rho$ under the \textbf{uniform traffic}. The TMS policy reconfigures the schedule $q = 10$ times within $qT$ time duration. The scheduling rate under either the TMS or PMW is equal to $1/T$, while under AMW is adapted to the traffic load intensity.}
\label{L_load_uniform}
\end{figure}

\begin{figure}
\includegraphics[width = 3.2in]{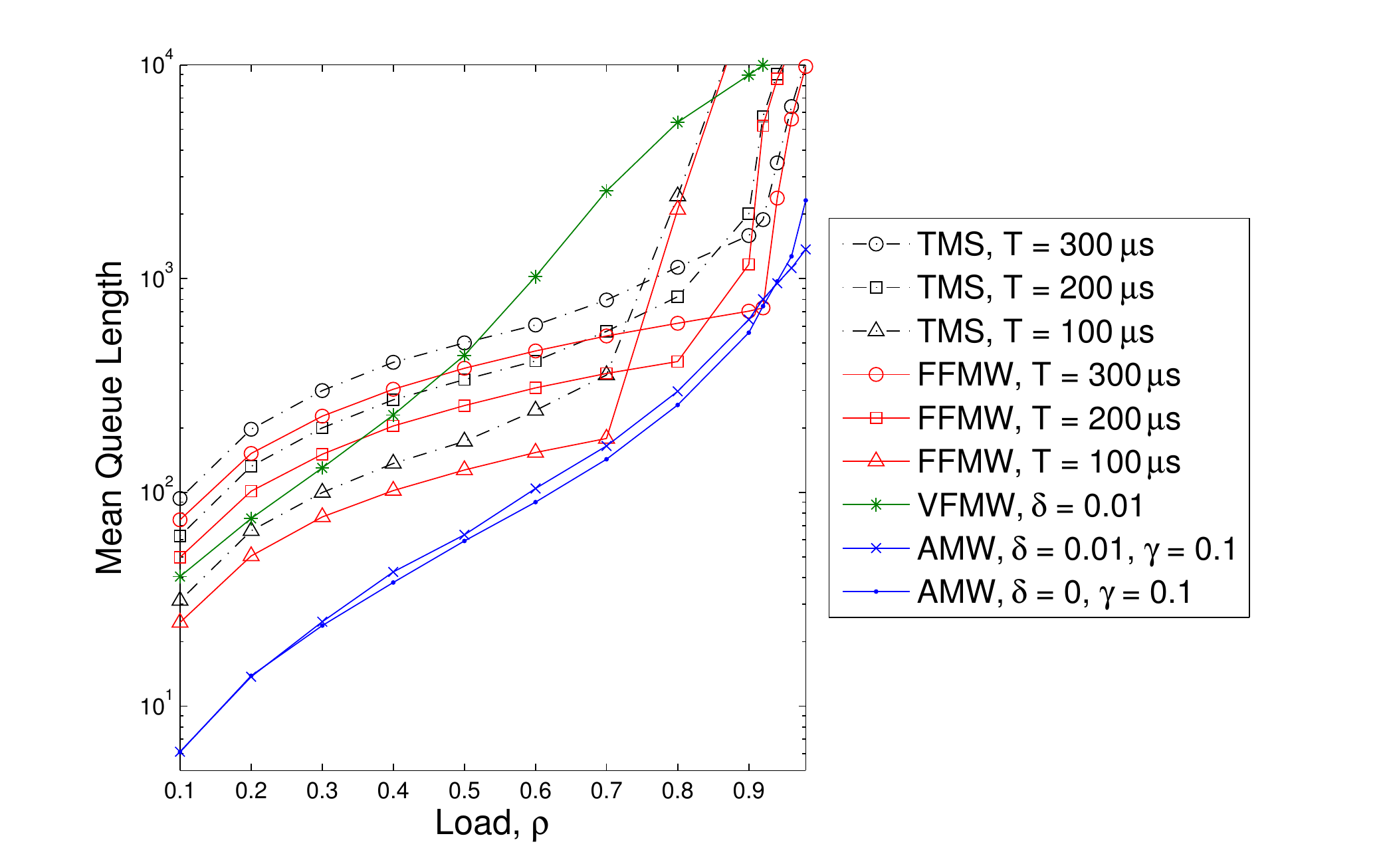}
\caption{Mean queue length versus traffic load $\rho$ under the \textbf{nonuniform traffic}. The TMS policy reconfigures the schedule $q = 10$ times within $qT$ time duration.}
\label{L_load_nonuniform}
\end{figure}

In this section we present simulation results for the AMW policy in Example 1, and compare them to the benchmark scheduling policies such as TMS, FFMW and VFMW. We also present comparison of the AMW policy against adaptive variants of lower complexity policies approximating the MaxWeight policy, as given in the examples in section \ref{adaptive_policies}.

The experiments are conducted with the simulator built for the REACToR switch in \cite{reactor}. The reconfiguration delay is $\Delta_r = 20\ \mu s$. In order to compare scheduling policies in optical switches, we cease the electronic switches in the hybrid switch in \cite{reactor} and only utilize the optical switches. We consider $N = 100$ ToR switches, and the network topology is assumed to be non-blocking. Therefore, the set of feasible schedules $\mathcal{S}$ is in fact the set of $N \times N$ permutation matrices. Each link has data bandwidth $B = 100$ Gbps, and the packets are of the same size $p = 1500$ bytes (each takes $0.12 \mu s$ for transmission). Each edge queue can store up to $1\times 10^5$ packets, and incoming packets are discarded when the queue is full.

The traffic is assumed to be admissible, i.e. $\rho(\boldsymbol{\lambda}) < 1$, while the load matrices $\boldsymbol{\lambda}$ are classified as the following types:
\begin{enumerate}
\item Uniform: $\lambda_{ij} = \rho / N, \ \forall 1 \leq i, j \leq N$.
\item Nonuniform: $\lambda_{ij} = \frac{\rho}{M} \sum_{m=1}^M \mathbf{P}_{ij}^m$ where $\mathbf{P}^m, m = 1, \dots, M \in \mathcal{P}$ are permutation matrices picked at random. The number $M$ determines the skewness of the load matrix. We set $M = 100$ here.
\end{enumerate}

\begin{figure}
\includegraphics[width = 3.2in]{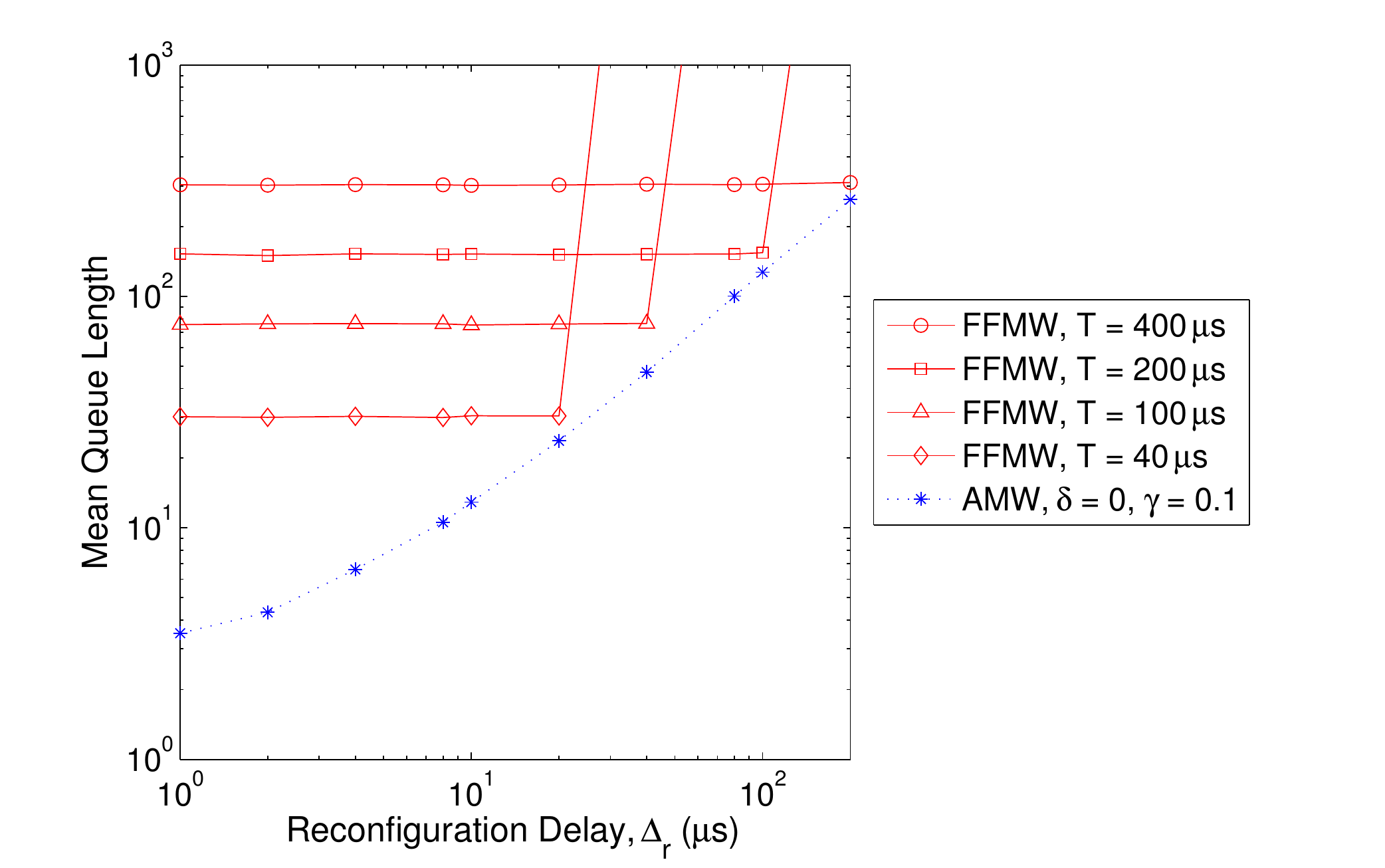}
\caption{Mean queue length versus the reconfiguration delay $\Delta_r$ under the nonuniform traffic. The traffic load is fixed as $\rho = 0.3$.
\newline
}
\label{night}
\end{figure}

The performance measure used is the mean edge queue length (averaged over queues and over time). Notice that the expected average delay of the entire network is linearly related to this quantity according to the Little's law.

In Figs \ref{L_load_uniform} and \ref{L_load_nonuniform}, we compare the scheduling policies described in section \ref{related} under the uniform and the nonuniform traffic, respectively. For TMS, we set the number of schedules used between two schedule computation time instances to be $Q = 10$. In Figs. \ref{L_load_uniform} and \ref{L_load_nonuniform} we can see that the TMS and FFMW perform comparably with the FFMW slightly outperforming the TMS under the same schedule reconfiguration rate $1/T = 10/W$. We note that under both the TMS and FFMW policies, the traffic loads they could stabilize are determined by the reconfiguration rate $1/T = 10/W$. In general, a smaller $T$ ($W$) value gives better delay performance at a fixed stable load, but choosing a smaller $T$ ($W$) value also decreases the maximum load that the TMS or FFMW policy could stabilize. We also note that the VFMW policy, which is shown to achieve the throughput optimality without requiring the knowledge of the arrival traffic, in practice, underperforms both the FFMW and TMS policies. On the other hand, the AMW policy ensures an improved performance over the FFMW and TMS policies and achieves throughput optimality.

\begin{figure}
\includegraphics[width = 3.0in]{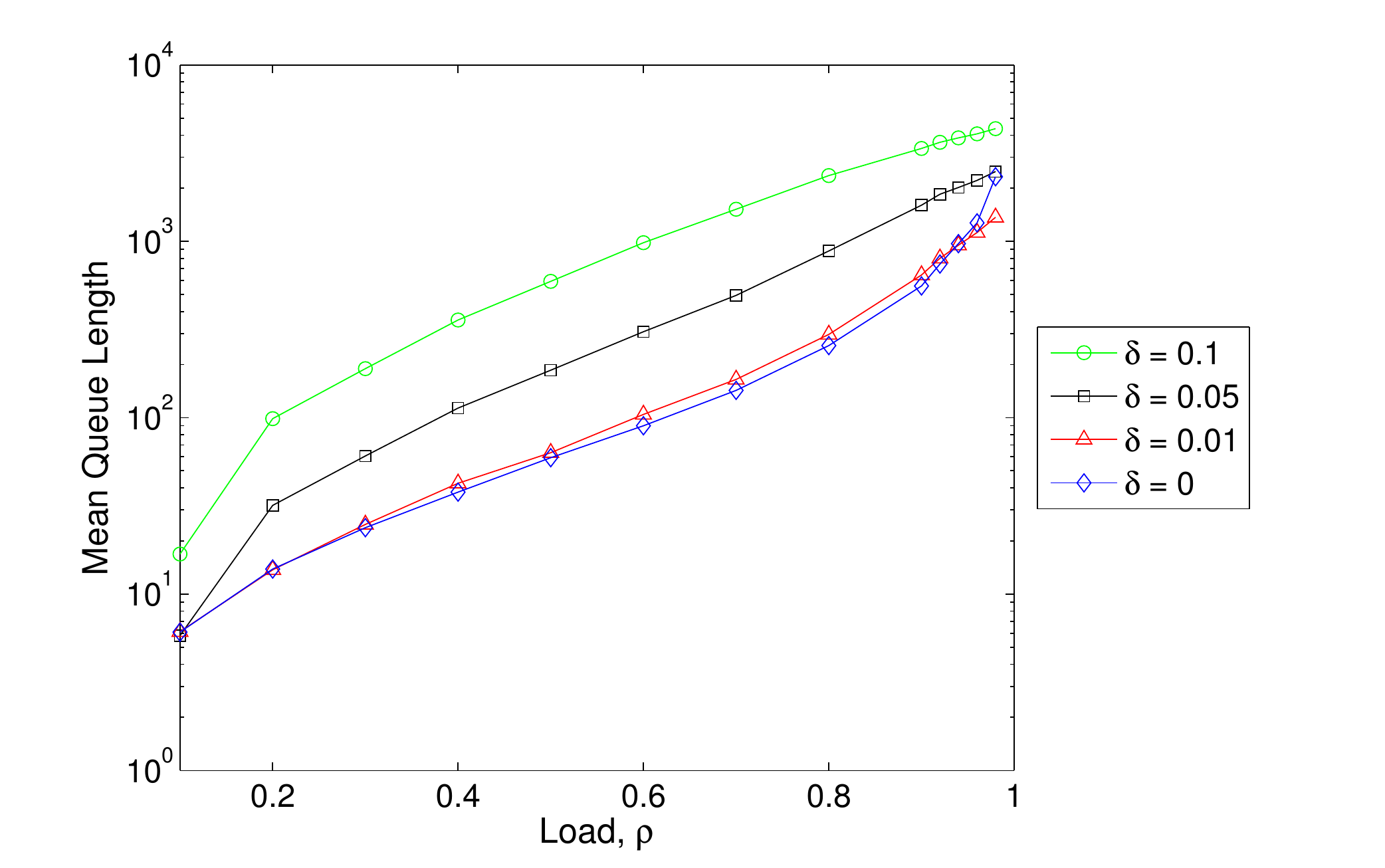}
\caption{Mean queue length versus traffic load for the AMW policy under sublinear exponent $\delta \in \{0, 0.01, 0.05, 0.1 \}$. The threshold ratio is fixed as $\gamma = 0.1$.}
\label{L_load_sublinear}
\end{figure}

\begin{figure}
\includegraphics[width = 3.0in]{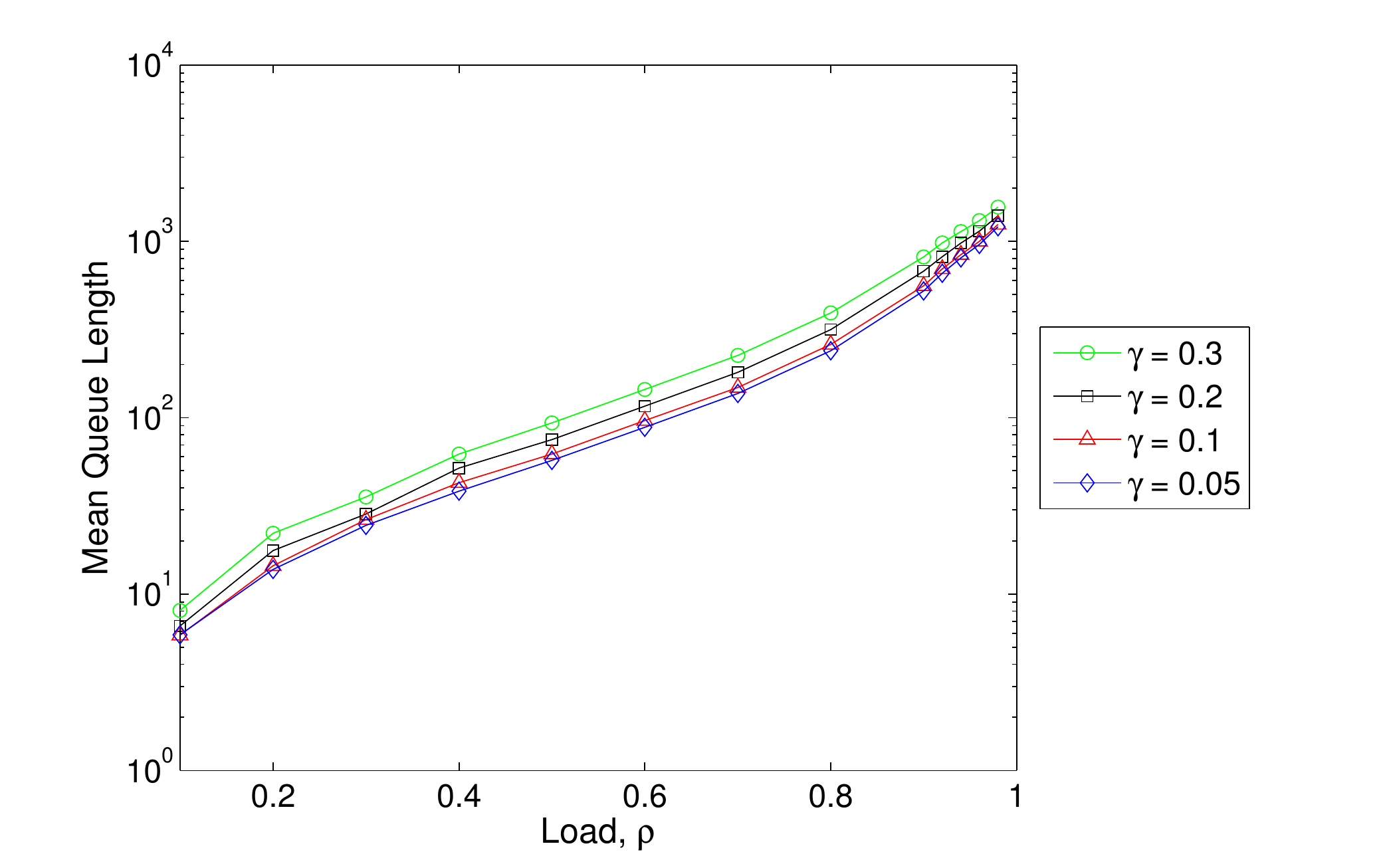}
\caption{Mean queue length versus traffic load for the AMW policy under threshold ratio $\gamma \in \{0.05, 0.1, 0.2, 0.3 \}$. The sublinear exponent is fixed as $\delta = 0.01$.}
\label{L_load_threshold}
\end{figure}

We now consider the effect of the reconfiguration delay $\Delta_r$ to the performance. In Fig. \ref{night}, we show the performance of the FFMW and AMW under various $\Delta_r$, while the traffic load is fixed as $\rho = 0.3$. We can see that the AMW outperforms the FFMW under each $\Delta_r$ value, regardless of the parameter selection of the FFMW policy. Although there exists an optimal schedule reconfiguration period $T$ of the FFMW policy that achieves comparable performance to the AMW policy at each $\Delta_r$, the choice of the optimal $T$ is dependent on the traffic load $\rho$. We can see that the performance of the AMW actually traces the optimal performance of the FFMW. This observation suggests that the adaptive strategy of the AMW in fact allows it to capture the optimal schedule reconfiguration rate based solely on the queue lengths information and no prior knowledge of the arrival statistics is required.

In Figs. \ref{L_load_sublinear} and \ref{L_load_threshold} we consider the effect of the parameter selection for the AMW policy. Recall from example \ref{adaptive_maxweight} that the $g$ function is selected as $g(x) = (1-\gamma)x^{1-\delta}$. Fig. \ref{L_load_sublinear} shows the performance for sublinear exponent $\delta \in \{0, 0.01, 0.05, 0.1 \}$, while the ratio threshold is fixed as $\gamma = 0.1$. Note that at any fixed traffic load $\rho$, the mean queue length becomes shorter when $\delta$ is smaller. Since our analysis (Corollary~\ref{coro1}) requires $\delta > 0$, we select $\delta = 0.01$ for the remaining simulations. Fig. \ref{L_load_threshold} presents the performance under the ratio threshold $\gamma \in \{0.05, 0.1, 0.2, 0.3 \}$, while the sublinear exponent is fixed as $\delta = 0.01$. Notice that the mean queue length decreases as $\gamma$ decreases, however, this sensitivity to the variable $\gamma$ seems to be fairly insignificant.

\begin{figure}
\psfrag{AMW}{test}
\includegraphics[width = 3.0in]{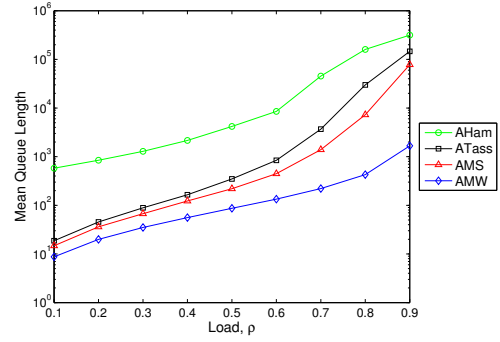}
\caption{Mean queue length versus traffic load for different adaptive policies.  The number of ToRs is $N = 8$ and $g(x) = (1-\gamma)x^{1-\delta}$ where $\gamma = 0.1, \delta = 0.01$.}
\label{L_load_linear_complexity}
\end{figure}

\begin{figure}
\includegraphics[width = 3.0in]{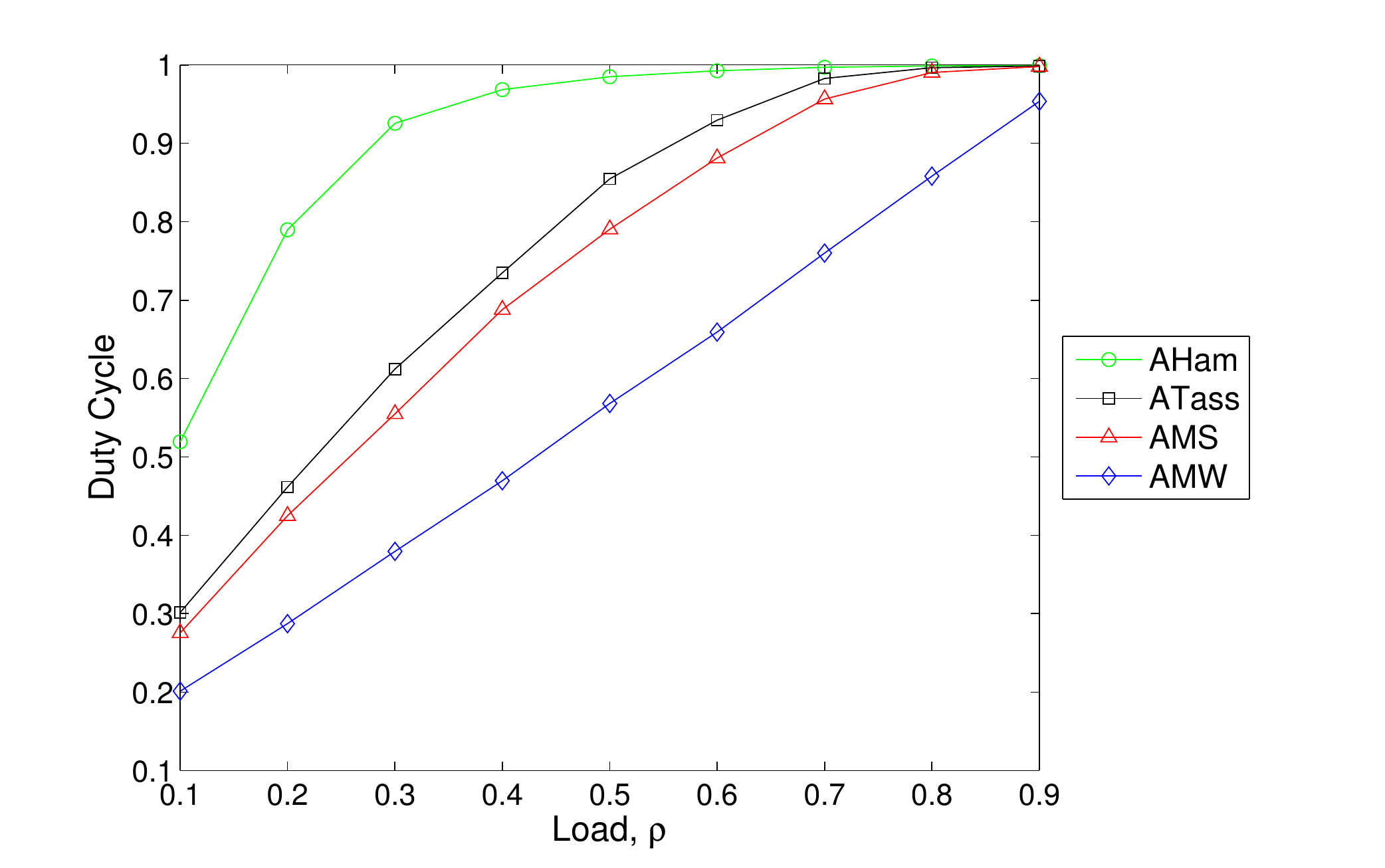}
\caption{Duty cycle versus traffic load for different adaptive policies. The number of ToRs is $N = 8$ and $g(x) = (1-\gamma)x^{1-\delta}$ where $\gamma = 0.1, \delta = 0.01$.}
\label{DC_load_linear_complexity}
\end{figure}

One of the shortcomings of the AMW policy is the complexity of computing the MaxWeight schedules. Next, we consider the adaptive variants of lower complexity scheduling policies introduced in Examples \ref{adaptive_tassiulas}-\ref{adaptive_maxsize}. In Fig. \ref{L_load_linear_complexity} we show the mean queue length versus traffic load for the Adaptive Hamiltonian (AHam), the Adaptive Tassiulas random policy (ATass), the Adaptive Maximum Size (AMS), and the AMW policies, under uniform traffic. Note that the delay performance of the lower complexity policies degrade drastically at large number of ToRs ($G \approx N!$ in Condition~\ref{strict_bound} for AHam and in Condition~\ref{expected_bound} for ATass), hence we set $N = 8$ for the simulations here. We also show the duty cycle of these policies in Fig. \ref{DC_load_linear_complexity}, where the duty cycle is defined as $DC \stackrel{\Delta}{=} 1 - \frac{\Delta_r}{\bbbe \{ T \}}$, while $\bbbe\{ T \}$ is the mean schedule duration. Note that a necessary condition for a scheduling policy to be throughput optimal under reconfiguration delay $\Delta_r > 0$ and traffic load $\rho$, is to satisfy $DC > \rho$, which is satisfied by all the scheduling policies shown here. Interestingly, due to the uniformity of the traffic, the AMS policy stabilizes the network and has comparable delay performance to the ATass policy, despite the fact that the AMS policy is not throughput optimal in general.

From Fig. \ref{L_load_linear_complexity}, we could see that the delay performance for the lower complexity policies are worse than the AMW policy. We may also observe that as the traffic load gets larger, the performance difference increases. This is because the schedule weight decreases in a slower rate (due to higher arrival rate), and it takes more time for the lower complexity policies to find a schedule with high enough weight. We could see in Fig. \ref{DC_load_linear_complexity} that the schedule durations become significantly long ($DC \rightarrow 1$) at high traffic loads. We also see that the AHam policy has the worst delay performance for all traffic loads, this is because the Hamiltonian walk only changes served queues in a time slot, and takes longer to find the next good schedule.

\section{Conclusion} \label{conclusion}
This paper considers the end-to-end scheduling problem in all-optical data center networks. The entire network is viewed as a generalized crossbar interconnect with nonzero schedule reconfiguration delay. Due to the schedule reconfiguration delay, many throughput optimal scheduling policies (under zero reconfiguration delay) in the literature could not be directly applied in this problem.

In this paper, we propose a general method to develop a class of scheduling policies for scheduling with nonzero reconfiguration delay, namely the $(g,f)$-adaptive variant policies: Given any Markov policy $\pi$, a weight function $f(\cdot)$, and a sublinear hysteresis function $g(\cdot)$, we construct a $(g,f)$-adaptive variant of $\pi$ which involves a weight comparison between the current schedule and the schedule generated by $\pi$, and reconfigure to the schedule generated by $\pi$ when it is ``significantly better.'' We show the throughput optimality of the $(g,f)$-adaptive variants of $\pi$ given the weight of schedule generated by $\pi$ is guaranteed to be close enough to the maximum weight (either in the deterministic or the expected sense).

\begin{figure}
\includegraphics[width = 3.0in]{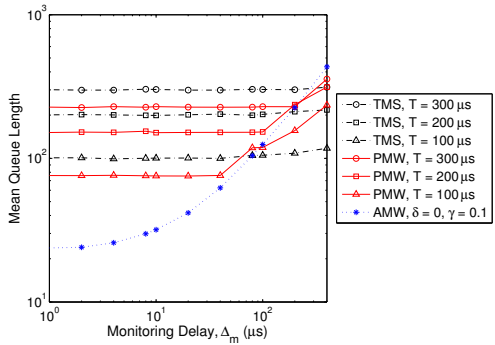}
\psfrag{xlabel}[c]{test}
\caption{Mean queue length versus monitoring delay $\Delta_m$ under nonuniform traffic. The traffic load is fixed as $\rho = 0.3$. The edge queue state is monitored/updated every microsecond, $t_{k+1}^M - t_k^M = 1 \mu$s for all $k$.}
\label{L_monitor}
\end{figure}

The proposed scheduling policies consider zero in-network buffer mainly due to the inherently bufferless nature of the optical circuits. It is interesting to note that the notion of edge-buffering and end-to-end scheduling has also been explored in the regime of electronic packet-switched data center network~\cite{fastpass} recently, in an effort to reduce buffering and congestion within the network. This suggests that our proposed scheduling policies can also be utilized in the context of electronic packet switches (with in-network buffering) in order to further reduce delay and improve performance.

In this work, we consider primarily the effect of the reconfiguration delay to the scheduling policies. In practice, however, monitoring (collecting queue length information) and schedule computation may also take a non-negligible time (denoted as $\Delta_m$ and $\Delta_c$, respectively), as suggested in \cite{Infocom15}. These delays can be viewed as incurring delay on the queue length information as in the Pipelined MaxWeight in example \ref{adaptive_pipeline_maxweight}. In other word, the stability of a $g$-adaptive policy is ensured even for $\Delta_m, \Delta_c > 0$. However, it is easy to see that our conclusion regarding the improved performance of dynamic policies over quasi-static policies does not hold when $\Delta_m$, $\Delta_c$ increase significantly. As an example, Fig. \ref{L_monitor} shows the delay performance under increased monitoring delay $\Delta_m > 0$. In the small $\Delta_m$ regime, the AMW policy achieves substantially better performance over the comparing schedulig policies. However, as $\Delta_m$ increases, the performance of the AMW policy sees a siginificant degradation. This observation motivates future research: 1) the development of low-delay ToR monitoring system~\cite{WiOpt}, 2) lower complexity scheduling policies with comparable delay performance, and 3) improvements to the AMW policy in order to increase the robustness with respect to the monitoring and computation delay.

\section*{Acknowledgment}

This work has been partially supported by L3 Communications, NSF Center for Integrated Access Networks (NSF Grant EEC-0812072), and NSF grant CNS-1329819 Event-Based Information Acquisition, Learning, and Control in High-Dimensional Cyber-Physical Systems.

\appendix

The following Foster Lyapunov Theorem is used to show the stability in the proofs listed in this appendix.

\begin{fact}[Foster-Lyapunov {\cite{Hajek}}]
\label{foster}
Given a system of edge queues $Q_{ij}, 1 \leq i,j\leq N$, with queue lengths $\mathbf{L}(t) = [L_{ij}(t)]$, which could be described by an irreducible discrete-time Markov chain (DTMC) on a countable state space $\mathcal{L}$. Let $f$, $g$ be two nonnegative functions on $\mathcal{L}$ such that $\forall \mathbf{L}(t_k) \in \mathcal{L}$,
\begin{align*}
&\bbbe \left[ V \left( \mathbf{L}(t_{k+1}) \right) - V(\mathbf{L}(t_k)) | \mathbf{L}(t_k) \right]\leq -f \big( \mathbf{L}(t_k) \big) + g \big( \mathbf{L}(t_k) \big)
\end{align*}
and suppose for some $\epsilon > 0$, the set $C = \{ \mathbf{L} \in \mathcal{L} : f(\mathbf{L}) <  g(\mathbf{L}) + \epsilon\}$ is finite, then the DTMC describing the queue length evolution is positive recurrent and we have
\begin{align*}
\lim_{k \rightarrow \infty} \bbbe \left[ f \big( \mathbf{L}(t_k) \big) \right] \leq \lim_{k \rightarrow \infty} \bbbe \left[ g \big( \mathbf{L}(t_k) \big) \right]
\end{align*}
\end{fact}

\subsection{Proof of Theorem \ref{Adaptive_Approximation}}
\label{theorem_adaptive_approximation_proof}

In the following proof, we use the quadratic Lyapunov function $V \left( \mathbf{L} \right) = \Big\langle \mathbf{L}, \mathbf{L} \Big\rangle = \sum_{i,j=1}^N L_{ij}^2$. 

Select $T$ such that $T > \frac{\Delta_r}{1 - \rho}$, and define the sequence of stopping times as $t_k = kT$.  Denote $\boldsymbol{\Delta}(t) = \mathbf{A}(t) - \mathbf{D}(t)$, and we may write the expected drift of the Lyapunov function as
\begin{align}
&\mathbb{E} \left[ V \left( \mathbf{L}(t_{k+1}) \right) - V(\mathbf{L}(t_{k})) | \mathbf{L}(t_k) \right]   \nonumber \\
=& \sum_{t=t_k}^{t_{k+1}-1} \mathbb{E} \left[ \Big\langle \mathbf{L}(t+1), \mathbf{L}(t+1) \Big\rangle - \Big\langle \mathbf{L}(t), \mathbf{L}(t) \Big\rangle \Big| \mathbf{L}(t_k) \right]   \nonumber\\
=& \sum_{t=t_k}^{t_{k+1}-1} \mathbb{E} \left[ \Big\langle 2 \mathbf{L}(t), \boldsymbol{\Delta }(t) \Big\rangle + \Big\langle \boldsymbol{\Delta}(t) , \boldsymbol{\Delta }(t) \Big\rangle \Big| L(t_k) \right]
\label{drift}
\end{align}

By the bound on the arrival process in Assumption \ref{assumption_traffic}, we have $\Delta_{ij}(t) \leq A_{\max}$, and thus $\Big\langle \boldsymbol{\Delta}(t), \boldsymbol{\Delta}(t) \Big\rangle \leq N^2 A_{\max}^2$.

It now remains to determine the first term in (\ref{drift}). For any queue $Q_{ij}$, at any time $t$, the departure $D_{ij}(t)$ differs from $S_{ij}(t)$ only if $L_{ij}(t) = 0$ or it is within the reconfiguration delay. Let $\mathbf{R}$ denote the set of time slots within which the network is doing reconfiguration, then we may write
\begin{align}
\Big\langle \mathbf{L}(t), \mathbf{D}(t) \Big\rangle &= \left\{ 
\begin{array}{ll} 
\Big\langle \mathbf{L}(t), \mathbf{S}(t) \Big\rangle, & \mbox{if } t \notin \mathbf{R} \\
0, & \mbox{if } t \in \mathbf{R} \\
\end{array} 
\right. \nonumber \\
&= \Big\langle \mathbf{L}(t), \mathbf{S}(t) \Big\rangle - \Big\langle \mathbf{L}(t), \mathbf{S}(t) \Big\rangle \mathds{1}_{\{ t \in \mathbf{R} \}}
\label{reconfiguration_overhead}
\end{align}
and thus
\begin{align}
&\bbbe \left[ \Big\langle \mathbf{L}(t), \boldsymbol{\Delta} (t) \Big\rangle \Big| \mathbf{L}(t_k) \right]   
= \bbbe \left[ \Big\langle \mathbf{L}(t), \mathbf{A}(t) - \mathbf{D}(t)  \Big\rangle  \Big| \mathbf{L}(t_k) \right] \nonumber \\
= & \Big\langle \mathbf{L}(t), \boldsymbol{\lambda} \Big\rangle -  \bbbe \left[ \Big\langle \mathbf{L}(t), \mathbf{S}(t)  \Big\rangle  \Big| \mathbf{L}(t_k) \right] \nonumber \\
&+ \bbbe \left[ \Big\langle \mathbf{L}(t), \mathbf{S}(t)  \Big\rangle \mathds{1}_{ \{t \in \mathbf{R} \}} \Big| \mathbf{L}(t_k) \right]  \nonumber \\
=& \bbbe \left[ \Big\langle \mathbf{L}(t), \boldsymbol{\lambda} - \mathbf{\Pi}^*(t) \Big\rangle + \Big\langle \mathbf{L}(t), \mathbf{\Pi}^*(t) - \mathbf{S}(t)  \Big\rangle  \Big| \mathbf{L}(t_k) \right] \nonumber \\
&+ \bbbe \left[ \Big\langle \mathbf{L}(t), \mathbf{S}(t)  \Big\rangle \mathds{1}_{ \{t \in \mathbf{R} \}} \Big| \mathbf{L}(t_k) \right] 
\label{L_deltaL_product}
\end{align}
where $\mathbf{\Pi}^*(t)$ is the MaxWeight schedule at time $t$. 

In \cite{VFMW} the capacity region has been characterized as the convex hull of the feasible schedules, that is
\[
\mathcal{C} = \left\{ \sum_{\mathbf{S} \in \mathcal{S}} \alpha_{\mathbf{S}} \mathbf{S} : \sum_{\mathbf{S} \in \mathcal{S}} \alpha_{\mathbf{S}} < 1, \ \alpha_{\mathbf{S}} \geq 0, \ \forall \mathbf{S} \in \mathcal{S} \right\}. 
\]
Hence by the definition of the traffic load $\rho$, we may write 
$\boldsymbol{\lambda} = \rho \sum_{i=1}^I \alpha_i \mathbf{S}_i $, where $\mathbf{S}_i \in \mathcal{S}$ and $0 \leq \alpha_i < 1$ for $i = 1, \dots, I$. By the definition of the MaxWeight schedule, we have 
\begin{align}
\Big\langle \mathbf{L}(t), \boldsymbol{\lambda} - \mathbf{\Pi}^*(t) \Big\rangle &\leq \rho \sum_{i=1}^{I} \alpha_i W^*(t) - W^*(t) \nonumber \\
&\leq -(1 - \rho) W^*(t)
\label{Birkhoff}
\end{align}

Also by the definition of the $g$-adaptive policy and the fact that the scheduling policy satisfies Condition~\ref{strict_bound}, we have 
\begin{align}
\Big\langle \mathbf{L}(t), \mathbf{\Pi}^*(t) - \mathbf{S}(t) \Big\rangle 
=& [W^*(t) - W^{\pi}(t)] + [W^{\pi}(t) - W^S(t)] \nonumber \\
\leq& G + g(W^{\pi}(t)) \leq G + g(W^*(t))
\label{weight_difference}
\end{align}

Apply (\ref{Birkhoff}) and (\ref{weight_difference}) into (\ref{L_deltaL_product}), we obtain
\begin{align}
\bbbe \left[ \Big\langle \mathbf{L}(t), \boldsymbol{\Delta} (t) \Big\rangle \Big| \mathbf{L}(t_k) \right]   
\leq& \bbbe \left[ -(1 - \rho) W^*(t) + g(W^*(t))  \right. \nonumber \\
&+ \left. G +  W^*(t) \mathds{1}_{ \{t \in \mathbf{R} \}} \Big| \mathbf{L}(t_k) \right] 
\end{align}

In order to achieve the stability, it is necessary to ensure that the reconfiguration does not happen too often. We make this statement formal with the following lemma:

\begin{lemma}
\label{lemma_frequency}
Given any fixed $T' > 0$ and a Markov policy $\pi$ satisfying Condition \ref{strict_bound} with weight function $f(x) = x$. Let $g(\cdot)$ be a sublinear and strictly increasing function, and let $M = g^{-1}(G + N(A_{max}+1) T') + NT'$. Suppose a reconfiguration occurs at time $t$ and $W^*(t) > M$, then no reconfiguration could occur in $[t+1, t+T']$.
\end{lemma}

The proof for lemma \ref{lemma_frequency} is given in appendix \ref{lemma_frequency_proof}. Note that lemma \ref{lemma_frequency} gives an upper bound on the frequency of schedule reconfiguration when the queue length is large, since the time between two schedule reconfigurations is larger than $T$ if $W^*(t) > M$. The value of $T'$ in lemma \ref{lemma_frequency} could be arbitrary in general, while in the following we set $T' = T$ (where $T > \frac{\Delta_r}{1-\rho}$ as selected earlier). Thus if $\forall t \in [t_k - \Delta_r, t_{k+1}]$: $W^*(t) > M = g^{-1}(G + N(A_{max}+1) T) + NT$, then at most one reconfiguration could occur and we have: 
\begin{align}
\sum_{t=t_k}^{t_{k+1}-1} \mathds{1}_{\{t\in\mathbf{R}\}} \leq \Delta_r
\end{align}

We thus have  $\forall \mathbf{L}(t_k) : W^*(t_k) > M$,
\begin{align}
&\bbbe \left[ V(\mathbf{L}(t_{k+1})) - V(\mathbf{L}(t_k)) \Big| \mathbf{L}(t_k) \right] \nonumber \\
\leq& \sum_{t=t_k}^{t_{k+1}-1} \bbbe \left[ -2(1 - \rho) W^*(t) + 2g \Big( W^*(t) \Big) + 2G \Big| \mathbf{L}(t_k) \right] \nonumber \\
&+ \sum_{t=t_k}^{t_{k+1}-1} \bbbe \left[  2W^*(t) \mathds{1}_{ \{t \in \mathbf{R} \}} \Big| \mathbf{L}(t_k) \right] + TN^2A_{\max}^2  \nonumber \\
\leq&  -2T(1 - \rho)  \left( W^*(t_k) - NT \right) + 2Tg\left( W^*(t_k) + NA_{\max}T \right)   \nonumber \\
&+ 2TG  +  2 \Delta_r \left( W^*(t_k) + NA_{\max}T \right) + TN^2A_{\max}^2  \nonumber \\
\leq& - \frac{2T}{N}(1 - \rho - \frac{\Delta_r}{T}) ||\mathbf{L}(t_k)|| + 2Tg\Big(||\mathbf{L}(t_k)|| + NA_{\max}T \Big) \nonumber \\
&+ 2T \Big(G + (1-\rho)NT + NA_{\max} \Delta_r \Big) + T N^2 A_{\max}^2
\label{drift1}
\end{align}
where the last inequality follows from  $||\mathbf{L}(t)|| \geq W^*(t) \geq \frac{1}{N}||\mathbf{L}(t)||$. Now since $g(\cdot)$ is a sublinear function, there exist constants $B, K < \infty$ and $\epsilon > 0$ such that for $||\mathbf{L}(t)|| > B$:
\[
\bbbe \left[ V(\mathbf{L}(t_{k+1})) - V(\mathbf{L}(t_k)) \Big| \mathbf{L}(t_k) \right] \leq - \epsilon ||\mathbf{L}(t_k)|| + K
\]
which by Fact \ref{foster}, implies $\lim\limits_{k \rightarrow \infty} \bbbe \left\{ ||\mathbf{L}(t_k)|| \right\} \leq K / \epsilon$ and thus the strong stability. Since it holds for any admissible traffic, the $g$-adaptive variant of $\pi$ achieves throughput optimality. 

Note that (\ref{drift1}) gives a bound on the drift for  $||\mathbf{L}(t_k)|| > NM$. On the other hand, for $||\mathbf{L}(t_k)|| \leq NM$, we also have a simple upper bound on the expected drift as
\begin{align}
&\bbbe \left[ V(\mathbf{L}(t_{k+1})) - V(\mathbf{L}(t_k)) \Big| \mathbf{L}(t_k) \right] \nonumber \\
\leq& \sum_{t = t_k}^{t_{k+1}-1} \bbbe \left[ 2\Big\langle \mathbf{L}(t), \boldsymbol{\Delta}(t) \Big\rangle + N^2 A_{\max}^2 \Big| \mathbf{L}(t) \right]  \nonumber \\
\leq & 2T||\mathbf{L}(t_k)|| A_{\max} + TN^2A_{\max}^2 \nonumber \\
\leq & 2TNM A_{\max} + TN^2A_{\max}^2 \nonumber \\
\leq & -\frac{2T}{N}(1-\rho-\frac{\Delta_r}{T})||\mathbf{L}(t_k)|| + 2T(1-\rho-\frac{\Delta_r}{T})M  \nonumber \\
&+ 2TNMA_{\max} + TN^2A_{\max}^2
\label{drift_inside}
\end{align}

Combining (\ref{drift1}) and (\ref{drift_inside}) we get a (rather loose) bound on the drift for all $\mathbf{L}(t_k)$ as
\begin{align*}
&\bbbe \left[ V(\mathbf{L}(t_{k+1})) - V(\mathbf{L}(t_k)) \Big| \mathbf{L}(t_k) \right] \nonumber \\
\leq& - \frac{2T}{N}(1 - \rho - \frac{\Delta_r}{T}) ||\mathbf{L}(t_k)|| + 
2T g\Big(||\mathbf{L}(t_k)|| + NA_{\max}T \Big) \nonumber \\
&+ 2T \Big(G + N (T + A_{\max} \Delta_r) + M(1 + NA_{\max}) + \frac{N^2 A_{\max}^2}{2} \Big) 
\end{align*}

From Fact \ref{foster} we obtain the following bound:
\begin{align*}
&\lim_{k \rightarrow \infty} \bbbe \Big[ ||\mathbf{L}(t_k)|| \Big]  \nonumber \\
\leq&  \frac{N}{1-\rho-\frac{\Delta_r}{T}}  \Big\{ \lim_{k \rightarrow \infty} \bbbe\left[ g\Big(||\mathbf{L}(t_k)|| + NA_{\max}T \Big) \right]  \nonumber \\
&+ \Big(G + N (T + A_{\max} \Delta_r) + M(1 + NA_{\max}) + \frac{N^2 A_{\max}^2}{2} \Big) \Big\} \\
\leq&  \frac{N}{1-\rho-\frac{\Delta_r}{T}}  \Big\{ g\Big( \lim_{k \rightarrow \infty} \bbbe\Big[ ||\mathbf{L}(t_k)|| \Big] + NA_{\max}T  \Big) \nonumber \\
&+ \Big(G + N (T + A_{\max} \Delta_r) + M(1 + NA_{\max}) + \frac{N^2 A_{\max}^2}{2} \Big) \Big\}
\end{align*}
where the last inequality follows from Jensen's inequality given the assumption that $g(\cdot)$ is a concave and continuous function. We then have the bound 
\begin{align}
\lim_{K \rightarrow \infty} \bbbe \left[ ||\mathbf{L}(t)|| \right] 
\leq \tilde{L}_T, \ \ \mbox{for each $T>\frac{\Delta_r}{1-\rho}$}
\end{align}
where $\tilde{L}_T$ satisfies $\tilde{L}_T = \frac{N}{1-\frac{\Delta_r}{T}-\rho} \Big\{  g(\tilde{L}_T + NA_{\max}T) + \Big(G + N (T + A_{\max} \Delta_r) + M(1 + NA_{\max}) + \frac{N^2 A_{\max}^2}{2} \Big) \Big\}$.

\subsection{Proof of Lemma \ref{lemma_frequency}}
\label{lemma_frequency_proof}

\begin{proof}
By the assumption, the schedule is reconfigured to $\mathbf{\Pi}(t)$ at time $t$, and according to Condition \ref{strict_bound} we have $W^{\pi}(t) = \Big\langle \mathbf{L}(t), \mathbf{\Pi}(t) \Big\rangle \geq W^*(t) -G$. We need to show that at any time $\tau \in [t+1, t+T]$, the weight of $\mathbf{\Pi}(t)$ is large enough so that no schedule reconfiguration occurs.

Since at most one packet could depart at each queue in a time slot, we have
\begin{align}
W(\tau) = \Big\langle \mathbf{L}(\tau), \mathbf{\Pi}(t) \Big\rangle &\geq W^{\pi}(t) - N(\tau - t) \nonumber \\
&\geq W^*(t) - G - NT 
\label{lemma1_eq1}
\end{align}

On the other hand, since the arrival at each queue is bounded by $A_{\max}$, we have an upper bound for the maximum weight:
\begin{align}
W^*(\tau) = \Big\langle \mathbf{L}(\tau), \mathbf{\Pi}^*(\tau) \Big\rangle &\leq W^*(t) + NA_{\max}(\tau - t) \nonumber \\
&\leq W^*(t) + NA_{\max}T
\label{lemma1_eq2}
\end{align}

From (\ref{lemma1_eq1}) and (\ref{lemma1_eq2}) we have: 
\begin{align*}
W^*(\tau)-W(\tau) \leq G + N(A_{\max} + 1)T
\end{align*}

Now since the maximum weight is at least the weight of the schedule $\mathbf{\Pi}^*(t)$, then from the bound on packet departure as in (\ref{lemma1_eq1}), we have a lower bound for the maximum weight:
\begin{align*}
W^*(\tau) \geq \Big\langle \mathbf{L}(\tau), \mathbf{\Pi}^*(t) \Big\rangle &\geq W^*(t) - N(\tau-t) \nonumber \\
&\geq W^*(t) - NT  
\end{align*}
Hence if $W^*(t) > M$, then $\forall \tau \in [t+1, t+T]$:
\begin{align*}
g\Big(W^*(\tau) \Big) > g(M - NT) &\geq G + N(A_{\max}+1)T \\
&\geq W^*(\tau) - W(\tau) 
\end{align*}
Since the weight difference does not exceed the threshold, no schedule reconfiguration could occur within $[t+1, t+T]$.
\end{proof}

\subsection{Proof of Theorem \ref{Adaptive_Expected}}

\begin{proof}
Notice that the policy $\pi$ (and thus its $g$-adaptive variant $\pi^g$) utilizes randomness in generating schedules, we take this into account when performing the drift analysis. We have the expected drift given by
\begin{align}
&\bbbe^{\pi^g} \left[ V(\mathbf{L}(t_{k+1})) - V(\mathbf{L}(t_k)) \Big| \mathbf{L}(t_k) \right] \nonumber \\
\leq& \sum_{t=t_k}^{t_{k+1}-1} \bbbe^{\pi^g} \left[ 2 \Big\langle \mathbf{L}(t), \mathbf{\Delta}(t) \Big\rangle  + N^2 A_{\max}^2 \Big| \mathbf{L}(t_k) \right]
\label{theorem2_derivation1}
\end{align}
where the expectation is taken with respect to the policy $\pi^g$.

Note that (\ref{reconfiguration_overhead}) in the proof of Theorem \ref{Adaptive_Approximation} remains the same in this case, we thus have
\begin{align}
& \bbbe^{\pi^g} \left[ \Big\langle \mathbf{L}(t), \boldsymbol{\Delta}(t) \Big\rangle \Big| \mathbf{L}(t_k)\right] \nonumber \\
= &  \bbbe^{\pi^g} \left[ \Big\langle \mathbf{L}(t), \boldsymbol{\lambda} - \mathbf{\Pi}^*(t) \Big\rangle + \Big\langle \mathbf{L}(t), \mathbf{\Pi}^*(t) - \mathbf{S}(t) \Big\rangle  \Big| \mathbf{L}(t_k) \right]  \nonumber \\
&+ \bbbe^{\pi^g} \left[ \Big\langle \mathbf{L}(t), \mathbf{S}(t) \Big\rangle \mathds{1}_{ \{t \in \mathbf{R} \}}  | \mathbf{L}(t_k)\right] \nonumber \\
\leq&  \bbbe^{\pi^g} \left[ - (1-\rho) W^*(t) + \Big\langle \mathbf{L}(t), \mathbf{\Pi}^*(t) - \mathbf{S}(t) \Big\rangle  \Big| \mathbf{L}(t_k) \right]  \nonumber \\
&+ \bbbe^{\pi^g} \left[ W^*(t)  \mathds{1}_{ \{t \in \mathbf{R} \} }  | \mathbf{L}(t_k)\right] 
\label{theorem2_derivation2}
\end{align}

By the construction of the $g$-adaptive variant and the fact that the policy $\pi$ satisfies Condition~\ref{expected_bound}, we have for any $t \geq t_k$:
\begin{align}
&\bbbe^{\pi^g} \left[  \Big\langle \mathbf{L}(t), \mathbf{\Pi}^*(t) - \mathbf{S}(t) \Big\rangle \Big| \mathbf{L}(t_k) \right] \nonumber \\
=& \bbbe^{\pi^g} \left[ \bbbe^{\pi} \left[ W^*(t) - W^{\pi}(t) \Big| \mathbf{L}(t) \right] \Big| \mathbf{L}(t_k) \right] \nonumber \\
&+ \bbbe^{\pi^g} \left[ \bbbe^{\pi} \left[ W^{\pi}(t) - W^{\pi^g}(t) \Big| \mathbf{L}(t) \right] \Big| \mathbf{L}(t_k) \right] \nonumber \\
\leq& G + \bbbe^{\pi^g} \left[ g(W^*(t)) \Big| \mathbf{L}(t_k) \right]
\label{theorem2_derivation4}
\end{align}

Similar to the proof in Theorem \ref{Adaptive_Approximation}, we need a bound on the rate of schedule reconfiguration. The following lemma works similarly as Lemma \ref{lemma_frequency} except it restricts the rate of reconfiguration in the average sense.
\begin{lemma}
\label{lemma_frequency_random}
Given any fixed $T' > 0$ and a scheduling policy $\pi$ satisfying Condition \ref{strict_bound} under the weight function $f(x) = x$. Let $g(x)$ be a sublinear and strictly increasing function, and let $h(x) = g(x - NT) - NA_{\max}T$. Then 
\begin{align*}
\bbbe^{\pi^g} \left[ \sum_{t = t_k}^{t_{k+1}-1} \mathds{1}_{\{t \in \mathbf{R} \}} \Big| \mathbf{L}(t_k) \right] \leq \Delta_r + \frac{TG}{h(W^*(t_k))}
\end{align*}
\end{lemma}
The proof of Lemma \ref{lemma_frequency_random} is given in Appendix \ref{lemma_frequency_random_proof} and follows the similar approach as in Lemma \ref{lemma_frequency}. 

Since $W^*(t_k) - NT \leq W^*(t) \leq W^*(t_k) + NA_{\max}T$ for all $t \in [t_k, t_{k+1}]$, applying Lemma \ref{lemma_frequency_random} we obtain

\begin{align*}
&\bbbe^{\pi^g} \left[ V(\mathbf{L}(t_{k+1})) - V(\mathbf{L}(t_k)) \Big| \mathbf{L}(t_k) \right] \nonumber \\
\leq& \sum_{t = t_k}^{t_{k+1}-1} \bbbe^{\pi^g} \left[ 2\Big\langle \mathbf{L}(t), \boldsymbol{\Delta}(t) \Big\rangle + N^2 A_{\max}^2 \Big| \mathbf{L}(t_k) \right] \nonumber \\
\leq& \sum_{t = t_k}^{t_{k+1}-1} \bbbe^{\pi^g} \left[ -2(1-\rho) W^*(t) + 2g(W^*(t)) + 2G \Big| \mathbf{L}(t_k)\right] \nonumber \\
&+ \sum_{t = t_k}^{t_{k+1}-1} \bbbe^{\pi^g} \left[ 2W^*(t) \mathds{1}_{\{ t \in \mathbf{R} \}} + N^2 A_{\max}^2 \Big| \mathbf{L}(t_k)\right] \nonumber \\
\leq& -2T \left( 1-\rho-\frac{\Delta_r}{T}-\frac{G}{h(W^*(t_k))} \right) W^*(t_k) \nonumber \\
&+ 2Tg\Big(W^*(t_k) + NA_{\max}T \Big) + TN^2 A_{\max}^2 \nonumber \\
&+ 2T \Big(G + N (T + A_{\max} \Delta_r) + \frac{NA_{\max}GT}{h(W^*(t_k))} + \frac{N^2A_{\max}^2}{2} \Big) 
\end{align*}

Since $T > \frac{\Delta_r}{1-\rho}$, we have $1 - \rho - \frac{\Delta_r}{T} > 0$, and we may select an arbitrary constant $\epsilon \in (0, 1 - \rho - \frac{\Delta_r}{T})$. Now since $\lim\limits_{x \rightarrow \infty} h(x) = \infty$, there exists a constant $M_1 < \infty$ such that $W^*(t_k) > M_1$ implies $1 - \rho - \frac{\Delta_r}{T} - \frac{G}{h(W^*(t_k))} > \epsilon$. Also, by the sublinearity of $g(\cdot)$, there exists a constant $M_2$ such that $W^*(t_k) > M_2$ implies $g(W^*(t_k)) < \frac{\epsilon}{3}W^*(t_k)$. Now let $M_3 = \frac{3}{\epsilon} \left( G + N (T + A_{\max} \Delta_r) + \frac{NA_{\max}GT}{h(W^*(t_k))} + \frac{N^2A_{\max}^2}{2} \right)$ and $B = \max\{M_1, M_2, M_3\}$, we have that if $W^*(t_k) > B$, then
\begin{align*}
&\bbbe^{\pi^g} \left[ V(\mathbf{L}(t_{k+1})) - V(\mathbf{L}(t_k)) \Big| \mathbf{L}(t_k) \right]  \\
<& - \epsilon W^*(t_k) + \frac{\epsilon}{3} W^*(t_k) + \frac{\epsilon}{3} W^*(t_k) = - \frac{\epsilon}{3} W^*(t_k)
\end{align*}

Recall that $ W^*(t) \geq \frac{1}{N} ||\mathbf{L}(t)||$ at any time $t$, we have that for $||\mathbf{L}(t_k)|| > NB:$
\begin{align}
\bbbe^{\pi^g} \left[ V(\mathbf{L}(t_{k+1})) - V(\mathbf{L}(t_k)) \Big| \mathbf{L}(t_k) \right] < - \frac{\epsilon}{3N} ||\mathbf{L}(t_k)||
\end{align}
then from Fact \ref{foster} we have that the network is strongly stable under any feasible traffic load, hence the $g$-adaptive variant of $\pi$ achieves throughput optimality. 

\end{proof}

\subsection{Proof of Lemma \ref{lemma_frequency_random}}
\label{lemma_frequency_random_proof}

\begin{proof}
Let $Z_{[s, t]}$ be the number of schedule reconfigurations within the time period $[s, t]$. Recall in lemma \ref{lemma_frequency} we showed that $Z_{[t_k, t_{k+1}]} \leq 1$ given $W^*(t_k)$ is large. Here we show a similar result in the probabilistic sense, in particular, we would derive a bound on $\Pr\{ Z_{[t_k, t_{k+1}]} \leq 1 \}$ in the following. 

Let $t'_k = \min\limits_{k: t_k^S \geq t} t_k^S$ denote the first reconfiguration instance after time $t_k$.
We may then write the event that exactly one schedule reconfiguration occurs in the interval $[t_k, t_{k+1}]$ as:
\begin{align*}
\{ Z_{[t_k, t_{k+1}]} = 1 \}  
=&  \bigcup_{\tau = t_k}^{t_{k+1}-1} \left\{ \{ \tau = t'_k \} \ \cap \ \{ Z_{[\tau, t_{k+1}]} = 0 \} \right\}. 
\end{align*}
Let $E_{\tau} = \left\{ \{ \tau = t'_k \} \ \cap \ \{ Z_{[\tau, t_{k+1}]} = 0 \} \right\}$. Note $E_{\tau}$ is the event that there is exactly one schedule reconfiguration within $[t_k, t_{k+1}]$ and it occurs at time $\tau$. The events $E_{\tau}$ are disjoint events for different $\tau$, hence we have that
\begin{align}
&\Pr \left\{ Z_{[t_k, t_{k+1}]} = 1 \Big| \mathbf{L}(t) \right\}  
=  \sum_{\tau = t_k}^{t_{k+1}-1} \Pr\left\{ E_{\tau} \Big| \mathbf{L}(t)\} \right\}  \nonumber \\
=& 
\begin{aligned}[t]
\sum_{\tau = t_k}^{t_{k+1}-1} \Pr & \left\{ Z_{[\tau, t_{k+1}]} = 0  \Big| \mathbf{L}(t), \tau = t'_k  \right\} \Pr \left\{ \tau = t'_k \Big| \mathbf{L}(t) \right\}
\end{aligned} \nonumber \\
\label{prob_one_reconfiguration}
\end{align}

Following the similar approach in Lemma \ref{lemma_frequency}, we show that if at the schedule reconfiguration instance $\tau$, the weight difference is small, then no reconfiguration could occur in the interval $[\tau, t_{k+1}]$. Specifically, let $h(W^*(t_k)) = g(W^*(t_k) - NT) - N(A_{\max}+1)T$, then if $W^*(\tau) - W_{\Pi^g (\tau)}(\tau) < h(W^*(t_k))$, we have for any $\tau' \in [\tau, t_{k+1}]$: 
\begin{align*}
&W^*(\tau') - W_{\Pi^g (t)}(\tau') \\
\leq& W^*(\tau) - W_{\Pi^g (t)}(\tau) + N (A_{\max}+1) (\tau' - \tau)\\
<& g(W^*(t_k) - NT) \stackrel{(a)}{\leq} g(W^*(\tau'))
\end{align*}
where $(a)$ follows from $W^*(\tau') \geq W^*(t) - N(\tau' - t) \geq W^*(t) - NT$ and $g$ strictly increasing. We thus have:
\begin{align*}
\Pr & \left\{ Z_{[\tau, t_{k+1}]} = 0 \Big| W^*(\tau) - W_{\Pi^g (\tau)}(\tau) < h(W^*(t_k)) \right\} = 1.
\end{align*}
We then obtain:
\begin{align}
&\Pr \left\{ Z_{[\tau, t_{k+1}]} = 0 \Big| \tau = t'_k \right\}  \nonumber \\
=& 
\begin{aligned}[t]
\Pr & \left\{ Z_{[\tau, t_{k+1}]} = 0 \Big| W^*(\tau) - W_{\Pi^g (\tau)}(\tau) \geq h(W^*(t)) \right\} \cdot \\
&\Pr \left\{ W^*(\tau) - W_{\Pi^g (\tau)}(\tau) \geq h(W^*(t)) \Big| \tau = t'_k \right\} 
\end{aligned}  \nonumber \\
&+ 
\begin{aligned}[t]
\Pr & \left\{ Z_{[\tau, t_{k+1}]} = 0 \Big| W^*(\tau) - W_{\Pi^g (\tau)}(\tau) < h(W^*(t)) \right\} \cdot \\
&\Pr \left\{ W^*(\tau) - W_{\Pi^g (\tau)}(\tau) < h(W^*(t)) \Big| \tau = t'_k \right\} \\
\end{aligned}  \nonumber \\
\geq& \Pr \left\{ W^*(\tau) - W_{\Pi^g (\tau)}(\tau) < h(W^*(t)) \Big| \tau =  t'_k \right\}  \nonumber \\
\stackrel{(b)}{\geq}& 1 - \frac{\bbbe^{\pi} \left[W^*(\tau) - W_{\Pi^g (\tau)}(\tau) \right] }{h(W^*(t))} 
\stackrel{(c)}{\geq} 1 - \frac{G}{h(W^*(t))}
\label{prob_zero_reconfiguration}
\end{align}
where $(b)$ follows from the conditional Markov's inequality, and $(c)$ follows from Condition~\ref{expected_bound}.

Notice that either there is no reconfigurations within $[t_k, t_{k+1}]$ or otherwise the first reconfiguration occurs at some $\tau \in [t_k, t_{k+1}]$, we thus have
\[
\Pr \left\{ Z_{[t_k, t_{k+1}]} = 0 \Big| \mathbf{L}(t) \right\} + \sum_{\tau = t_k}^{t_{k+1}-1} \Pr \left\{ \tau = t'_k \Big| \mathbf{L}(t) \right\} = 1
\]
hence by (\ref{prob_one_reconfiguration}) and (\ref{prob_zero_reconfiguration}) we have that
\begin{align}
&\Pr \left\{ Z_{[t_k, t_{k+1}]} \leq 1 \Big| \mathbf{L}(t) \right\}  \nonumber \\
=& \Pr \left\{ Z_{[t_k, t_{k+1}]} = 0 \Big| \mathbf{L}(t) \right\} \nonumber \\
&+ 
\sum_{\tau = t_k}^{t_{k+1}-1} \Pr \left\{ Z_{[\tau, t_{k+1}]} = 0  \Big| \mathbf{L}(t), \tau =  t'_k  \right\} \Pr \left\{ \tau = t'_k \Big| \mathbf{L}(t) \right\}  \nonumber \\
\geq& \Pr \left\{ Z_{[t_k, t_{k+1}]} = 0 \Big| \mathbf{L}(t) \right\}  \nonumber \\
&+ \sum_{\tau = t_k}^{t_{k+1}-1} \left( 1 - \frac{G}{h(W^*(t))} \right) \Pr \left\{ \tau = t'_k \Big| \mathbf{L}(t) \right\}  \nonumber \\
\geq&  1 - \frac{G}{h(W^*(t))}
\label{prob_less_than_one}
\end{align}

With the following bound which is obvious by definition:
\begin{align*}
\sum_{t = t_k}^{t_{k+1}-1} \mathds{1}_{\{t \in \mathbf{R}\}} &\leq \left\{ 
\begin{array}{ll} 
\Delta_r, & \mbox{if } Z_{[t_k, t_{k+1}]} \leq 1 \\
T, & \mbox{if } Z_{[t_k, t_{k+1}]} > 1 \\
\end{array} 
\right. ,
\end{align*}
along with (\ref{prob_less_than_one}), we then have the bound on the expected schedule reconfiguration delay within the interval $[t_k, t_{k+1}]$:
\begin{align*}
\bbbe^{\pi^g} \left[  \sum_{t = t_k}^{t_{k+1}-1} \mathds{1}_{\{t \in \mathbf{R}\}}  \Big| \mathbf{L}(t_k) \right] \leq \Delta_r + \frac{TG}{h(W^*(t_k))}
\end{align*}
\end{proof}

\bibliographystyle{ieeetr}
\bibliography{reference.bib}

\end{document}